%% file: 00_main.tex
\documentclass[format=acmsmall, review=false, nonacm]{acmart}
\usepackage{acm-ec-26}
\makeatletter
\let\@ACM@checkaffil\relax
\makeatother

\usepackage{booktabs} % For formal tables
\usepackage[ruled]{algorithm2e} % For algorithms
\usepackage{algorithmic}

\SetAlFnt{\small}
\SetAlCapFnt{\small}
\SetAlCapNameFnt{\small}
\SetAlCapHSkip{0pt}
\IncMargin{-\parindent}
\setcitestyle{authoryear}
\AtEndPreamble{%
    \theoremstyle{acmdefinition}
    \newtheorem{assumption}[theorem]{Assumption}}

\AtEndPreamble{%
    \theoremstyle{acmdefinition}
    \newtheorem{remark}[theorem]{Remark}}

\title[The Architecture of Illusion: Network Opacity and Strategic Escalation]{The Architecture of Illusion: Network Opacity and Strategic Escalation}

\author{Raman Ebrahimi}
\affiliation{%
  \institution{University of California, San Diego}
  \department{Electrical and Computer Engineering}
  \city{La Jolla}
  \state{CA}
  \country{USA}
}
\email{raman@ucsd.edu}

\author{Sepehr Ilami}
\affiliation{%
  \institution{Northeastern University}
  \department{Mechanical and Industrial Engineering}
  \city{Boston}
  \state{MA}
  \country{USA}
}
\email{ilami.a@northeastern.edu}

\author{Babak Heydari}
\affiliation{%
  \institution{Northeastern University}
  \department{Mechanical and Industrial Engineering}
  \city{Boston}
  \state{MA}
  \country{USA}
}
\email{b.heydari@northeastern.edu}

\author{Isabel Trevino}
\affiliation{%
  \institution{University of California, San Diego}
  \department{Economics}
  \city{La Jolla}
  \state{CA}
  \country{USA}
}
\email{itrevino@ucsd.edu}

\author{Massimo Franceschetti}
\affiliation{%
  \institution{University of California, San Diego}
  \department{Electrical and Computer Engineering}
  \city{La Jolla}
  \state{CA}
  \country{USA}
}
\email{mfranceschetti@ucsd.edu}

\begin{abstract}

Standard models of bounded rationality typically assume agents either possess accurate knowledge of the population's reasoning abilities (Cognitive Hierarchy) or hold dogmatic, degenerate beliefs (Level-$k$). We introduce the ``Connected Minds'' model, which unifies these frameworks by integrating iterative reasoning with a parameterized network bias. We posit that agents do not observe the global population; rather, they observe a sample biased by their network position, governed by a locality parameter $p$ representing algorithmic ranking, social homophily, or information disclosure. We show that this parameter acts as a continuous bridge: the model collapses to the myopic Level-$k$ recursion as networks become opaque ($p \to 0$) and recovers the standard Cognitive Hierarchy model under full transparency ($p=1$).

Theoretically, we establish that network opacity induces a \emph{Sophisticated Bias}, causing agents to systematically overestimate the cognitive depth of their opponents while preserving the log-concavity of belief distributions. This makes $p$ an actionable lever: a planner or platform can tune transparency, globally or by segment (a personalized $p_k$), to shape equilibrium behavior. From a mechanism design perspective, we derive the \emph{Escalation Principle}: in games of strategic complements, restricting information can maximize aggregate effort by trapping agents in echo chambers where they compete against hallucinated, high-sophistication peers. Conversely, we identify a \emph{Transparency Reversal} for coordination games, where maximizing network visibility is required to minimize variance and stabilize outcomes. Our results suggest that network topology functions as a cognitive zoom lens, determining whether agents behave as local imitators or global optimizers.

\end{abstract}



\begin{document}

\begin{titlepage}

\maketitle

% Optionally include a table of contents
% \setcounter{tocdepth}{2} % adjust to 1 if desired
% \tableofcontents

\end{titlepage}

% Paper body
\input{01_introduction}
\input{02_model}
\input{03_theoretical}
\input{04_mechanism_design}

\section{Conclusion}
The \emph{Connected Minds} framework formalizes the structural intuition that strategic behavior is shaped as much by \textit{who we see} as by \textit{how we think}. By treating network structure as a cognitive filter, we demonstrate that deviations from Nash Equilibrium are not solely the product of limited computational capacity (finite $k$), but also of limited observational capacity (finite $p$).

Theoretically, the model serves as a unification bridge. We established that the locality parameter $p$ interpolates continuously between the beliefs of the Level-$k$ model and the calibrated expectations of the Cognitive Hierarchy model. Crucially, we prove that this bridge is structurally stable: network-induced biases preserve the log-concavity of belief distributions. This ensures that agents, regardless of how warped their local view, perceive a coherent, unimodal world, ruling out chaos in favor of systematic, predictable shifts in behavior. Our analysis uncovers a fundamental limit to individual sophistication. Through the \emph{Poisson-Shift Convergence} theorem, we prove that network position eventually trumps cognitive depth. In opaque networks, even agents with infinite recursive reasoning capabilities ($k \to \infty$) fail to converge to the truth. Instead, they stabilize at a perceived sophistication level of $\tau/p$, optimizing against a phantom population of peers. This offers a structural explanation for persistent bubbles: agents may be locally hyper-rational yet globally blinded, systematically mistaking the echo of their own sophistication for the signal of the market.

From a policy perspective, we introduce the role of the \emph{Cognitive Designer}, revealing a stark duality in information architecture. We derive the \emph{Escalation Principle}, showing that in competitive environments (e.g., crowdsourcing, R\&D), a designer should engineer opacity to trap agents in high-performance echo chambers. Conversely, we identify the \emph{Transparency Reversal} for coordination games, where maximizing $p$ is strictly optimal to minimize variance. Furthermore, we highlight a distributional consequence: opacity acts as a \emph{sophistication tax}. While mid-level agents are the most \textit{elastic} to changes in transparency (per Theorem~\ref{thm:hierarchyofsensitivity}), opacity disproportionately \textit{distorts} the steady-state beliefs of the most sophisticated agents, who suffer the largest absolute divergence between their perceived ($\tau/p$) and actual ($\tau$) competitive environments.

\paragraph{Limitations and Future Research}
Our framework opens several avenues for future inquiry, driven by the limitations of the current analysis.
First, our treatment of $p$ is a reduced-form summary of informational locality. While the discounting kernel $p^{k-h}$ captures the fading visibility of distant types, it abstracts away from specific graph topologies. Future work should simulate the model on explicit adjacency matrices to map graph statistics, such as modularity and clustering coefficients, to the effective $p$ derived in our theoretical results.
Second, the current model is static. We analyze equilibrium behavior under fixed beliefs, but we do not model the co-evolution of network structure and strategic learning. A critical extension is to investigate dynamic settings: do agents infer $p$ from realized payoffs, or does the homophily of their network trap them in self-confirming incorrect beliefs?
Third, the model presents an identification challenge for empiricists. Since observational reach ($p$) and cognitive depth ($k$) can be observationally equivalent in static action data (confounded as $\tau/p$), testing these predictions requires experimental designs that exogenously manipulate visibility while holding incentives fixed.

Ultimately, \emph{Connected Minds} shifts the locus of behavioral intervention. Correcting strategic errors requires more than merely educating agents; it requires architectural reform. To dissolve polarized echo chambers or burst speculative bubbles, we must alter the topology of information flow, rendering the invisible layers of the hierarchy visible.

% In the interest of anonymization, please do not include acknowledgements in your submission.
%
%\begin{acks}
%
%	The authors would like to thank Dr. Maura Turolla of Telecom
%	Italia for providing specifications about the application scenario.
%
%	The work is supported by the \grantsponsor{GS501100001809}{National
%		Natural Science Foundation of
%		China}{http://dx.doi.org/10.13039/501100001809} under Grant
%	No.:~\grantnum{GS501100001809}{61273304\_a}
%	and~\grantnum[http://www.nnsf.cn/youngscientsts]{GS501100001809}{Young
%		Scientsts' Support Program}.
%
%
%\end{acks}

% Bibliography
\bibliographystyle{ACM-Reference-Format}
\bibliography{sample-bibliography}

\newpage

% Appendix
\appendix
\input{041_mechanism_design}
\input{91_numerical}
\input{92_proofs}

\end{document}

%% file: 01_introduction.tex
\section{Introduction}
The investigation of human decision-making has historically oscillated between two poles: the hyper-rational optimization of \textit{homo economicus} and the heuristic-driven, bounded rationality of real-world agents. Many models, such as Prospect Theory~\cite{tversky1979prospect} or limited search~\cite{simon1955behavioral}, explain systematic departures from rational choice, and laboratory experiments~\cite{nagel1995unraveling,smith1962experimental} document how such departures arise in strategic environments. A common simplifying assumption across much of this work is \emph{independent belief formation}: a decision-maker computes a best response as if their beliefs about others are formed separately from the social context in which those beliefs are acquired.

In socially embedded settings, however, best responses are shaped by everyday interaction, which quietly alters what we believe others believe. Strategic environments also frequently contain \emph{hierarchies of sophistication}: some participants devote substantial effort to anticipate others, while many do not. Yet even highly sophisticated actors rarely observe the full distribution of strategic depths directly; they infer it from limited neighborhoods of interaction and filtered signals.

This paper argues that independent belief formation is empirically untenable in such contexts: agents do not view the strategic landscape from an Archimedean point; they view it from a specific node within a complex topology. Topological positioning triggers socio-cognitive mechanisms where belief revision is no longer purely logical, but socially fabricated within an agent's immediate interactive field~\cite{burns2001socio}. Relatedly, the ``social bubble'' phenomenon~\cite{pariser2011filter,sunstein2018republic} suggests that an agent's perception of society is systematically skewed by their local connections. We therefore study how network structure can systematically distort perceived population composition, even when agents are strategically sophisticated.

We formalize this intuition through the \emph{Connected Minds} model. We embed the Cognitive Hierarchy (CH) framework within a network structure and introduce a continuous locality parameter $p \in (0,1]$ that quantifies the transparency of the social graph. The parameter bridges myopic agents ($p \to 0$), who effectively observe only a narrow local neighborhood, and global agents ($p=1$), who perceive the true hierarchy of the entire population as assumed in standard CH formulations. By relaxing the implicit accuracy of beliefs about population composition, the model isolates how topology acts as a cognitive filter.

Our analysis yields two main contributions. First, we establish the \emph{Poisson-Shift Convergence Theorem}. Standard epistemic intuitions suggest that as reasoning depth increases ($k \to \infty$), beliefs should converge to the truth. We show that in a biased network ($p<1$), sophistication does not eliminate distortion: beliefs stabilize at a phantom distribution characterized by a shifted parameter $\tau/p$. In other words, position can dominate cognition; a highly sophisticated agent embedded in a dense local cluster can confidently optimize for a systematically misperceived population.

Second, we extend the framework to a normative \emph{Mechanism Design} setting. We introduce a ``Cognitive Designer'' who manipulates network transparency to achieve social goals. We identify a policy trade-off governed by game structure. In games of strategic complements (e.g., technology adoption or effort), we derive an \emph{Escalation Principle}: opacity ($p<1$) can maximize aggregate effort by inducing agents to compete against overestimated local competition. Conversely, in coordination settings (e.g., the Beauty Contest), we identify a \emph{Transparency Reversal}, where full transparency is required to minimize belief variance and reduce payoff inequality. Together, these results imply that optimal information architecture depends on whether the designer prioritizes aggregate output or stability.

The remainder of the paper situates these contributions relative to prior work on bounded rationality, misspecified learning, and strategic interaction on networks, before developing the model and welfare implications in detail.

\paragraph{Contributions.}
We introduce a behavioral framework in which agents' iterative reasoning is filtered by network-mediated observation.
Our main contributions are:
(i) a parsimonious \emph{network locality} parameter $p\in(0,1]$ that endogenizes the perceived distribution of cognitive types via a tractable kernel $g_k(\cdot\mid p)$;
(ii) general comparative statics showing that changes in $p$ induce monotone shifts in beliefs while preserving regularity (e.g., log-concavity), enabling sharp predictions across broad classes of games;
(iii) an asymptotic characterization for Poisson populations in which perceived sophistication converges to an \emph{effective} mean $\tau/p$, generating stable, topology-induced miscalibration even as cognitive depth grows;
and (iv) a normative design problem in which a planner/platform chooses information architectures (captured by $p$) to trade off escalation, coordination, and inequality.

\section{Related Work}
Our work lies at the intersection of bounded rationality in strategic settings, learning with misspecified beliefs, and network economics. We connect cognitive hierarchy models to networked belief formation through a locality parameter~$p$, and study how modulating information structure (network transparency) shapes system-level welfare.

\paragraph{Bounded Rationality and Iterative Reasoning.}
Our model builds on the Level-$k$ and Cognitive Hierarchy (CH) frameworks, developed by \citet{nagel1995unraveling}, \citet{stahl1995players}, and \citet{costa2001cognition} to explain behavior in beauty contest games, a concept originally described by \citet{keynes1937general}. These frameworks were subsequently generalized in CH to account for responses to a distribution of lower-level types~\citep{camerer2004cognitive,wright2014level}. A key premise supporting iterative reasoning is overconfidence (the ``better-than-average'' effect)~\citep{hoffrage2022overconfidence}. While Level-$k$ models typically assume agents hold degenerate beliefs that all others belong to a single lower type (usually $k-1$), CH posits that agents hold accurate beliefs about the conditional distribution of lower-level types. However, empirical work finds systematic inconsistencies in such beliefs~\citep{fragiadiakis2017testing,crawford2007level,agranov2017stochastic,agranov2024beliefs}. Our contribution is to endogenize these distortions as a function of network position and information locality.

\paragraph{Belief Formation, Misspecification, and Social Learning.}
A related literature studies inference and learning when agents hold incorrect models of the environment. Small misspecifications can generate persistent biases and long-run incorrect learning~\cite{bohren2016informational}. In our setting, misspecification is topological: agents treat locally observed behavior as representative of the global population, with~$p$ governing the degree of locality. This is closely related to correlation neglect, where agents fail to account for shared information sources among neighbors' actions~\cite{eyster2005cursed,enke2019correlation}. Overconfidence may persist because it can be evolutionarily stable even when inaccurate~\cite{johnson2011evolution}, and it is empirically linked to important distortions in high-stakes environments such as corporate investment~\cite{malmendier2005ceo}.

\paragraph{Networks, Interaction, and Polarization.}
Our approach relates to work on strategic interaction on networks~\cite{galeotti2010network} and to agent-based approaches to opinion dynamics~\cite{banisch2019opinion}, but differs by studying non-equilibrium iterative reasoning on networks. The structural bias induced by locality also connects to echo chambers and algorithmic reinforcement of homophily in online platforms~\cite{cinelli2021echo}.

\paragraph{Mechanism Design, Information Design, and Welfare.}
Our work complements research on mechanism design with boundedly rational agents~\cite{de2019level,rivas2015mechanism} by focusing on the information structure rather than payoff rules: a ``Cognitive Designer'' modulates network transparency via~$p$ to target welfare objectives. Our ``Optimal Opacity'' results also resonate with work in information economics where transparency can reduce welfare~\cite{ostrizek2024acquisition}, and with platform settings in which agents anticipate algorithms and peer behavior~\cite{zhu2025look}. Finally, our welfare analysis relates to behavioral evidence on networked decision-making~\cite{kearns2012behavioral,papi2020ordered} and to normative considerations about efficiency--equity trade-offs~\cite{ismail2020one}.

%% file: 02_model.tex
\section{Connected Minds Model}
In this section, we formally define the strategic environment and the probabilistic mechanism by which network topology distorts belief formation. The model integrates the iterative reasoning of Cognitive Hierarchy theory with a parameterized network bias.

\subsection{Game Setup and Reasoning Hierarchy}
We consider a population of $N$ agents interacting in a strategic setting. Let $s_i^j$ denote player $i$'s $j$\textsuperscript{th} strategy, assuming player $i$ has a finite set of $m_i$ strategies.

The hierarchy of reasoning is anchored by Level-0 ($L_0$) thinkers. These agents are non-strategic and do not engage in iterative reasoning. Consistent with the literature, we assume level-0 thinkers choose according to a probability distribution not derived from maximizing a payoff function against an opponent. Specifically, we utilize a uniform distribution for level-0 players, implying that every available strategy is equally likely to be chosen: $\mathbb{P}_0(s_i^j) = \frac{1}{m_i}, \forall i \in {n_0}$, where $n_0$ denotes the set of level-0 players in the population. Agents with reasoning level $k > 0$ are strategic and operate under a model of bounded rationality. Their decision-making process decomposes into two stages: 1) \emph{Belief formation}: The agent forms a subjective belief about the distribution of reasoning levels in the opposing population, and 2) \emph{Best Response}: The agent calculates the expected action of the population based on their belief and selects a strategy that maximizes their expected payoff.

\subsection{Network-Based Belief Formation}
In the standard Cognitive Hierarchy (CH) model \cite{camerer2004cognitive}, a level-$k$ player assumes that the population consists of players with levels $0, 1, \dots, k-1$. Crucially, the standard model assumes that the level-$k$ player knows the true relative frequencies of these lower levels. For instance, if the true population follows a Poisson distribution, the level-$k$ player uses a truncated, normalized Poisson distribution as their belief.

The Connected Minds model relaxes this assumption of accuracy. We posit that agents interact on a network $G=(V, E)$, represented by an adjacency matrix $A$ where $A_{ij}=1$ if players $i$ and $j$ are connected. Agents do not observe the global population; \textit{they observe a sample biased by their network position.}

We define the belief of a level-$k$ agent about the distribution of agents on other levels ($h$) as $g_k(h)$. We maintain the standard assumption of overconfidence: A level-$k$ player acts as if they are the most sophisticated player in the game. They assign zero probability to the existence of any player with reasoning level $h \ge k$. Thus, $g_k(h) = 0$ for all $h \ge k$.

However, the belief $g_k(h)$ regarding lower levels ($h < k$) is distorted. We assume that social networks exhibit homophily regarding cognitive sophistication—agents are more likely to be connected to others with similar reasoning capabilities. Specifically, for nontrivial levels $k>1$ and under biased/local observation ($p<1$), a level-$k$ thinker is (weakly) more likely to be connected to a near-peer level-$(k\!-\!1)$ player than to a more distant low-level player (e.g., level $0$); at $p=1$ this homophily bias vanishes and the connection kernel is uniform across all $h<k$.

We model this bias using a network locality parameter $p \in (0, 1]$. The subjective belief $g_k(h)$ is defined as the true distribution $f(h)$ weighted by a kernel that decays with the "cognitive distance" between the observer ($k$) and the observed ($h$):
\begin{align}
    g_k(h) = \frac{p^{k-h}f(h)}{\sum_{l=0}^{k-1}p^{k-l}f(l)}, \quad \forall h \in {0, \dots, k-1} \label{eq:prob-dist}
\end{align}
Here, $f(h)$ is the objective, true probability mass function (PMF) of reasoning levels in the population (e.g., Poisson). 
$p \in (0, 1]$\footnote{We exclude $p=0$ because \eqref{eq:prob-dist} is defined via normalization. However, the limit is well-defined: as $p\to 0$, $g_k$ concentrates its mass entirely on the highest observable type $h=k-1$. In this limit, the belief structure mimics the classic Level-$k$ recursion (where Level $k$ responds only to Level $k-1$), though the aggregate population behavior may differ due to the exact mass distribution.} represents the reach or transparency of the network. The term $p^{k-h}$ represents the connection probability or relative visibility of a level-$h$ agent to a level-$k$ agent. Equation~\eqref{eq:prob-dist} encapsulates the agent's predisposition. As the difference $k-h$ increases (i.e., looking further down the hierarchy), the weight $p^{k-h}$ decreases (for $p < 1$). In Section~\ref{subsec:measure-found} we generalize this definition.

\subsection{Microfoundations of the Locality Parameter}\label{subsec:microfoundations}
While we model $p$ as a reduced-form weighting parameter, it admits rigorous microfoundations in both network topology and algorithmic information sorting. We provide two distinct structural interpretations that yield the exponential kernel $p^{k-h}$ naturally.

\paragraph{Interpretation 1: Cognitive Homophily (Stochastic Block Model).} 
Consider the population structure as a specialized Stochastic Block Model (SBM) or a Soft Random Geometric Graph. Let the cognitive hierarchy define a latent space where the distance between two agents $i$ (level $k$) and $j$ (level $h$) is $d_{ij} = |k-h|$.
If social ties are formed based on cognitive proximity (homophily), the probability of a link forming between these agents decays with their cognitive distance. Specifically, if the connection probability follows a geometric decay:
\begin{equation}
    \mathbb{P}(A_{ij}=1 \mid k, h) \propto p^{|k-h|}
\end{equation}
While this topology exposes agents to both lower ($h < k$) and higher ($h > k$) types, the agent's cognitive bound $k$ acts as a censoring filter. Consistent with the axioms of Cognitive Hierarchy theory, we assume the agent cannot distinguish types $h \ge k$ from noise or treats them as lower-level deviations. Consequently, \eqref{eq:prob-dist} represents the \emph{conditional} distribution of observed types, truncated to the agent's comprehensible support $\mathcal{H}_k = \{0, \dots, k-1\}$.
Then, the parameter $p$ functions as the \textit{decay factor of the homophily kernel}. A lower $p$ implies a society strictly stratified by intellect (strong homophily), while $p \to 1$ implies a society where cognitive distance does not inhibit connection (weak homophily).

\paragraph{Interpretation 2: Algorithmic Ranking and Attention.} 
In digital environments, $p$ represents the \textit{ranking decay} of a content curation algorithm. Consider a social network (e.g., LinkedIn or Twitter/X) feed algorithm that sorts potential peers or posts. Users have a finite attention span.
Suppose a platform's recommendation engine clusters peers by similarity—prioritizing level $k-1$ peers (rank distance 1), followed by level $k-2$ peers (rank distance 2), and so on. We model the user's attention using Rank-Biased Precision \cite{moffat2008rank}, where the probability of examining content at rank distance $r$ decays as $p^r$.
In this construction, the parameter $p$ generates the \textit{attention kernel}. When this attention filter operates on a population where types are distributed according to $f(h)$, the effective sample observed by the user is the product of the attention probability and the population frequency: $g_k(h) \propto p^{k-h}f(h)$.
Here, $p$ is an engineering parameter: the persistence of the user's scroll or the diversity setting of the recommender system.

\paragraph{Interpretation 3: Information Diffusion (Random Walk).} 
Alternatively, consider information passing through a network via a Random Walk with Restart (PageRank) \cite{tong2006fast,page1999pagerank}. If a level-$k$ agent seeks information, and that search propagates down the hierarchy with a continuation probability $p$ at each step (and a dampening/restart probability $1-p$), the stationary distribution of visited nodes relative to the source $k$ follows the geometric sequence $p^{k-h}$. Here, $p$ represents the \textit{conductivity} or transparency of the network to information flow.

\subsection{Parameter $p$ as a Sophistication Switch}\label{parameter_p}
The interpretation of $p$ is central to the model's insights. It acts as a continuous bridge between two distinct classes of behavioral models:
\begin{enumerate}
    \item $p \to 0$ (Myopic/Local Connections): In this limit, the weighting term $p^{k-h}$ decays extremely rapidly as $h$ moves away from $k$. The belief mass becomes overwhelmingly concentrated on the level immediately below the agent, $h = k-1$. Mathematically, $g_k(k-1) \to 1$. This simplifies the framework to the classic Level-$k$ Model, where a Level-2 player believes everyone is Level-1, and a Level-3 player believes everyone is Level-2. This corresponds to a player embedded in a tight, insular cluster (an echo chamber) who only interacts with peers slightly less sophisticated than themselves.
    \item $p = 1$ (Sophisticated/Global Connections): In this limit, $p^{k-h} = 1$ for all $h$. The weighting kernel is uniform. The belief becomes $g_k(h) = \frac{f(h)}{\sum_{l=0}^{k-1} f(l)}$. This recovers the standard Cognitive Hierarchy (CH) Model, where the agent uses the correct conditional distribution of lower-level types. This corresponds to a player with a global view, fully connected to the network, who accurately perceives the heterogeneity of the population.
\end{enumerate}
By varying $p$, the Connected Minds model allows us to explain why different reasoning models appear to hold in different contexts. Level-$k$ behavior is not necessarily a distinct cognitive process from CH behavior; rather, they are manifestations of the same cognitive process operating under different network topologies (low $p$ vs. high $p$).

\begin{figure}[ht]
    \includegraphics[width=\textwidth]{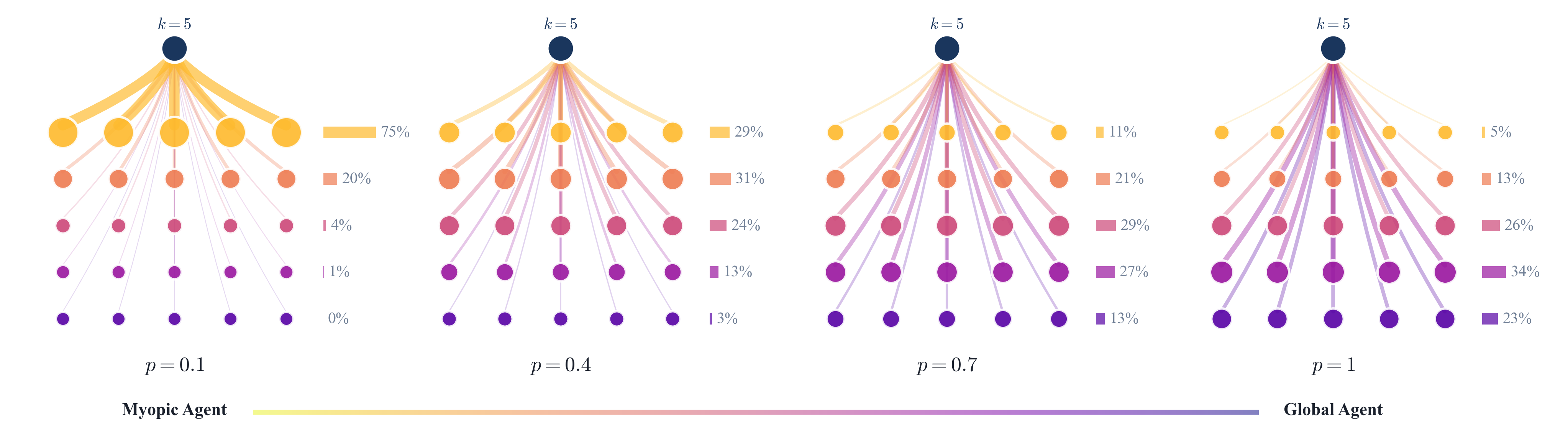}
    \caption{Example with k=5 (reasoning levels) showing how filtering happens. Agents only see neighbors with specific probability weights. Layout optimized to show connection spread clearly as we go from myopic agent (left) to global agent (right).}
    \label{fig:effects_of_p_how}
\end{figure}

%% file: 03_theoretical.tex
\section{Theoretical Properties}\label{sec:theory}
In this section, we analyze the structural and behavioral implications of the Connected Minds model. We begin by situating the framework within the classic Beauty Contest Game \cite{bosch2002one} to illustrate how the locality parameter $p$ shifts traditional level-$k$ predictions. We then broaden our scope to derive measure-theoretic foundations, establishing that the cognitive biases induced by network structure preserve the fundamental stability of belief distributions while systematically shifting their mass.

\subsection{The Beauty Contest Game: A Benchmarking Case}
The \emph{Guessing 2/3 of the Average} game, or Beauty Contest game, serves as the standard Drosophila for investigating depth of reasoning in experimental economics. For this game, a player wants their chosen number to be as close as possible to the target value, which is 2/3 of the average of all numbers chosen.\footnote{Originally, it is $p\times Average$ but the most used $p$ is $2/3$.} Let the average of all chosen numbers be $\bar s = \frac{1}{N}\sum_{j=1}^N s_j$. It is well known that the unique Nash equilibrium for the classic beauty contest game is $s_j=0$ for all $j$. 

In this framework, a player's action depends on their level of reasoning, $k$. A level-$k$ player simulates the thinking of all players they believe exist (i.e., those with levels $l<k$) to form an expectation of the average guess and then plays a best response. The model is anchored by level-0 players, who are non-strategic. They don't reason about the game at all and are typically assumed to choose a number randomly. If choices are from $[0,100]$, their average action is: $s_0=50$. 

A player with level $k>0$ determines their action, $s_k$, by doing the following steps: (1) They calculate the predicted actions of all lower-level players ($s_1, \ldots, s_{k-1}$); (2) They form a belief about the proportion of these lower-level players in the population using the belief function $g_k(h)$ in \eqref{eq:prob-dist}; (3) They calculate their expected global average, $E_k[\bar s]$; (4) They play their best response, which is $s_k = \frac{2}{3}E_k[\bar s]$. 

Writing out this mathematically, we have:
\begin{align}
    s_k = \big( \frac{2}{3} \big) \frac{\sum_{h=0}^{k-1} p^{k-h} f(h) s_h}{\sum_{l=0}^{k-1} p^{k-l} f(l)}
\end{align}

This model replaces the strong rationality assumptions of Nash Equilibrium with a new model of boundedly rational, iterative thinking. The most interesting finding from our model lies in the interpretation of the parameter $p$. 
As discussed in Section ~\ref{parameter_p}, $p$ interpolates between Level-k and CH; here we show how it shifts predicted actions: (i) When $p\to 0$ (Myopic/Local Connections): The term $p^{k-h}$ heavily penalizes looking at distant levels. A player's belief becomes overwhelmingly focused on the level immediately below them ($h=k-1$). The model simplifies to the classic, simple level-k model where $s_k\simeq \frac{2}{3}s_{k-1}$. This represents a player whose connections are very ``local,'' leading them to make simple, iterated assumptions about others. (ii) When $p=1$ (Sophisticated/Global Connections): The parameter $p$ has no effect. A player averages over all lower levels they can conceive of, weighted by their (truncated) population frequency. This represents a player with a more ``global'' view, who tries to perform a more complex cognitive average.

Our model proposes that the structure of a player's social network can determine their effective reasoning style. Players in insular, echo-chamber-like networks might behave like low-$p$ agents, while players in diverse, well-connected networks might behave like high-$p$ agents. This allows us to explain why different reasoning models might be appropriate in different contexts, linking network topology directly to cognitive strategy. We start by looking into two simple scenarios where the real distribution of reasoning levels is derived from (1) Geometric and (2) Poisson distributions.

\begin{example}\label{example:geometric}
    Imagine a population in which $f(h)$ follows a Geometric distribution. A Geometric distribution gives us the form $f(h; q) = (1-q)^h q$. 
    By substituting this into \eqref{eq:prob-dist} we get:
    \begin{align}
        g_k(h) = \frac{p^{k-h}(1-q)^h q}{\sum_{l=0}^{k-1}p^{k-l}(1-q)^l q}=\frac{p^{-h}(1-q)^h }{\sum_{l=0}^{k-1}p^{-l}(1-q)^l}
    \end{align}
    and by defining $r = \frac{1-q}{p}$ we get:
    \begin{align}
        g_k(h) = \frac{r^h}{\sum_{l=0}^{k-1}r^l}
    \end{align}
    As we can see above, the belief follows a \textit{truncated geometric distribution} with parameter $r = \frac{1-q}{p}$. Here, $1-q$ is the true probability of thinking one level deeper. Since $p<1$, the perceived likelihood ratio of deeper reasoning $r$ is always greater than the true probability $1-q$; notably, if the network bias is strong ($p < 1-q$), then $r$ exceeds 1, implying the agent perceives deeper opponents as more likely than shallower ones. 
    
    To put this formally: if the true probability of a player attaining at least level $k+1$ given they have attained level $k$ is $\pi_\text{true}=1-q$, then a level-$k$ player's perceived continuation rate for players in their model is $\pi_\text{perc}= \frac{1-q}{p}$. Thus, for any $p<1$, it holds that $\pi_\text{perc}> \pi_\text{true}$.

\end{example}
Having $f(h)$ follow a Geometric distribution implies that reasoning is a sequential process with a constant cost. Imagine players climbing a ladder of reasoning. At each rung, there is a fixed probability $q$ that they stop and a probability $1-q$ that they invest the cognitive energy to climb to the next rung. This describes a population with a higher number of simpler thinkers, where the count of ``survivors'' at each level of sophistication drops by a constant factor. Example~\ref{example:geometric} tells us that when we have a population where the distribution of reasoning levels is Geometric, such that, as described above, the true probability of any player engaging in an additional level of reasoning is constant. Therefore, a player whose beliefs are formed through a local sampling of their connections (characterized by a locality parameter $p<1$) will systematically overestimate the propensity of others to engage in deep reasoning. 

\begin{example}\label{example:poisson}
    Now, consider a population where the distribution of reasoning levels is Poisson with a true mean sophistication of $\tau$. The mean of the level-k player's subjective belief distribution over lower levels, $\mu_{g_k}$, is strictly greater than the true conditional mean of the population they are reasoning about, $\mu_{f_{<k}}=\mathbb{E}_f[h|h<k]$, i.e. $\mu_{g_k} > \mu_{f_{<k}}$.
    
    Given a Poisson $f(h)$ we will have $f(h;\tau)=\frac{e^{-\tau}\tau^h}{h!}$ and by substituting in \eqref{eq:prob-dist} we get:
    \begin{align}
        g_k(h) = \frac{p^{k-h}\frac{e^{-\tau}\tau^h}{h!}}{\sum_{l=0}^{k-1}p^{k-l}\frac{e^{-\tau}\tau^l}{l!}}
    \end{align}
    The term $e^{-\tau}$ cancels out, and the denominator is simply a normalization constant, let's call it $Z_k$, that ensures the probabilities sum to 1. The belief is therefore proportional to the numerator:
    \begin{align}
        g_k(h)\propto \frac{p^{k-h}\tau^h}{h!} = p^k\frac{(\tau/p)^h}{h!}
    \end{align}
    Since $p\in[0,1]$, the perceived mean $\tau/p$ is always greater than or equal to the true mean $\tau$. This means that \textit{a player systematically overestimates the average sophistication of the population they are perceiving}. They are surrounded by higher-level thinkers and mistakenly believe the rest of the world is just as sophisticated, ignoring the silent majority of low-level thinkers.
\end{example}

Having $f(h)$ follow a Poisson distribution implies that reasoning ability is a naturally occurring trait. Like height or weight, most people are clustered around an average level of sophistication ($\tau$), with very few people at the extremes. It suggests a single, underlying cognitive process common to the population. This is the standard and most empirically supported assumption in the cognitive hierarchy literature, representing a ``normal'' population of thinkers. In this scenario, Example~\ref{example:poisson} tells us that the local sampling process creates an illusion that the average player is more sophisticated than they actually are. The arguments of the examples above are illustrated in Figure~\ref{fig:fig6_examples} and we formalize the generalized result in Theorem~\ref{thm:poissonshiftconvg}.

\begin{figure}[ht]
    \centering
    \includegraphics[width=0.95\linewidth]{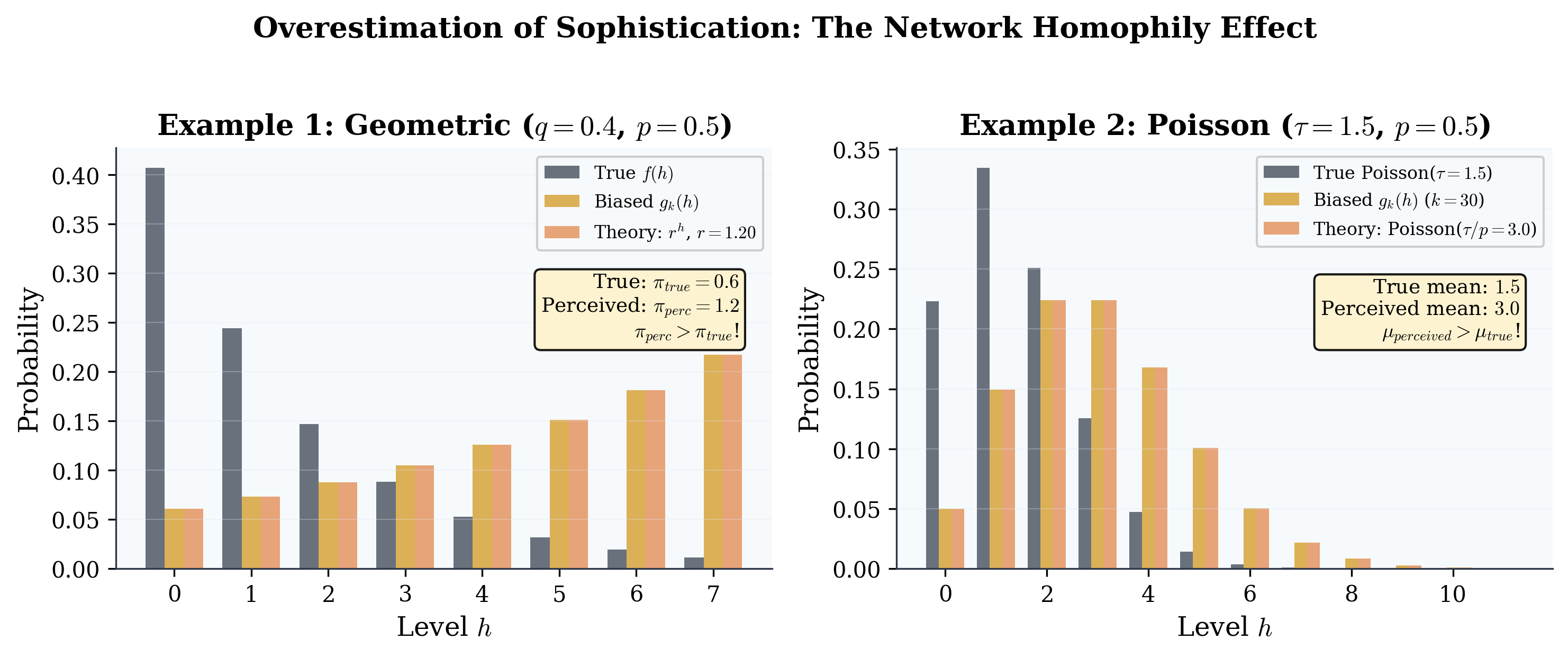}
    \caption{The impact of biased network sampling showing the difference between true and perceived distributions of levels for true Geometric (left) and Poisson (right) distributions.}
    \label{fig:fig6_examples}
\end{figure}

\subsection{Measure-Theoretic Foundations of Biased Beliefs}\label{subsec:measure-found}
We now rigorously define the probabilistic environment and derive the fundamental structural properties of the belief distribution $g_k(h)$. We treat the hierarchy of reasoning types as a discrete state space $\Omega = \mathbb{N}_0$ and analyze the probability measure induced by the parametric weighting kernel $K(k, h; p) = p^{k-h}$.

\subsubsection{The Parametric Family of Belief Distributions}
Let $\mathbb{N}_0 = \{0, 1, 2, \dots\}$ denote the set of possible cognitive levels (or ``types''). Let $f: \mathbb{N}_0 \to \mathbb{R}_{>0}$ be a probability mass function (PMF) representing the objective or true distribution of cognitive capacities in the population. The most common empirical specification for $f$ is the Poisson distribution, $f(h) = e^{-\tau}\tau^h/h!$, although our results hold for a broader class of distributions \cite{chen2014cognitive}.

A level-$k$ agent, denoted $L_k$, operates under the structural constraint that they can only conceptualize and best-respond to agents of lower reasoning depth, specifically the set $\mathcal{H}_k = \{0, 1, \dots, k-1\}$. The agent's subjective belief $g_k(h)$ over this support is defined as the normalized product of the objective prior and the exponential weighting kernel:
\begin{align}\label{eq:agentsubjbelief}
    g_k(h; p) = \frac{1}{Z_k(p)} p^{k-h} f(h) \cdot \mathbb{I}(h < k)
\end{align}
where $\mathbb{I}(\cdot)$ is the indicator function and $Z_k(p)$ is the partition function (or normalization constant) required to ensure $\sum_{h=0}^{k-1} g_k(h; p) = 1$. 

\begin{definition}
    (Cognitive Bias Regimes). The parameter $p$ defines the regime of cognitive bias:
    \begin{itemize}
        \item Standard CH ($p=1$): The weighting kernel is uniform ($1^{k-h}=1$). The agent uses the correct conditional probability $g_k(h) = f(h) / F(k-1)$. This assumes the agent perfectly understands the relative frequencies of lower types \cite{fragiadiakis2017testing,camerer2004cognitive}.
        \item Sophisticated Bias ($p < 1$): The kernel is strictly increasing in $h$. The agent believes the relevant population is concentrated just below their own level. This reflects ``Neuroticism'' or over-attribution of strategic depth to opponents \cite{bohren2016informational}.
        \item Naive Bias: Within $p \in (0, 1]$, naive bias corresponds to the kernel $p^h$, which is strictly decreasing in $h$ and therefore overweights low types. Equivalently, one can keep the kernel $p^{k-h}$ and allow $p>1$ \cite{bohren2016informational}. 
    \end{itemize}
\end{definition} 
We will focus our attention on the sophisticated bias case, as this is closer to reality \cite{ross1977false,marks1987ten}. However, our main results also hold for the naive bias, and we briefly discuss the intuition behind some results for the naive bias case. 

\subsubsection{Exponential Tilting and the Preservation of Structure}
The transformation $f(h) \mapsto g_k(h)$ represents a truncated exponential tilting of the probability measure. A critical question for the stability of game-theoretic equilibria is whether this transformation preserves the ``nice'' properties of the prior, specifically log-concavity. Log-concave distributions are central to economics as they ensure the uniqueness of optima and satisfy the Monotone Likelihood Ratio Property (MLRP) \cite{chen2014cognitive, saumard2014log}. 

\begin{theorem}\label{thm:logconcavity}
    (Preservation of Log-Concavity under Biased Tilting).
    Let the objective distribution of types $f(h)$ be log-concave on its support $\mathbb{N}_0$. Then, for any bias parameter $p \in (0, 1]$ and any level $k \ge 1$, the subjective belief distribution $g_k(h; p)$ is log-concave with respect to $h$ (for both kernels $p^{(k-h)}$ and $p^h$).
\end{theorem}
% \begin{proof}
%     A discrete distribution $q(h)$ is log-concave if the sequence satisfies the condition $q(h)^2 \ge q(h-1)q(h+1)$ for all $h$ in the interior of the support. This is equivalent to the concavity of the function $\psi(h) = \ln q(h)$, characterized by the second difference operator $\Delta^2 \psi(h) \le 0$.
%     Consider the log-belief function for the Connected Minds model:
%     $$\psi_k(h) = \ln g_k(h) = \ln ( \frac{p^k}{Z_k(p)} p^{-h} f(h) ) = \ln(p^k) - \ln Z_k(p) - h \ln p + \ln f(h)$$
    
%     To determine concavity, we compute the second difference with respect to $h$:
%     $$\Delta^2 \psi_k(h) = \psi_k(h+1) - 2\psi_k(h) + \psi_k(h-1)$$
    
%     Substituting the expansion:
%     $$\Delta^2 \psi_k(h) = [-(h+1)\ln p + \ln f(h+1)] - 2[-h\ln p + \ln f(h)] + [-(h-1)\ln p + \ln f(h-1)]$$
%     (Note that the constant terms $\ln p^k$ and $\ln Z_k(p)$ vanish upon differencing).
%     Grouping the linear terms involving $\ln p$:
%     $$-(h+1)\ln p + 2h\ln p - (h-1)\ln p = \ln p (-h - 1 + 2h - h + 1) = 0$$
    
%     The linear tilt $p^{-h}$ contributes a linear term to the log-density, which has a vanishing second derivative (or difference).
%     Thus, the curvature is determined entirely by the prior: 
%     $$\Delta^2 \psi_k(h) = \ln f(h+1) - 2\ln f(h) + \ln f(h-1) = \Delta^2 (\ln f(h))$$
%     By the assumption that $f$ is log-concave, $\Delta^2 (\ln f(h)) \le 0$. Therefore, $\Delta^2 \psi_k(h) \le 0$, proving that $g_k(h)$ is log-concave.
% \end{proof}

\emph{Insight and Implication:} 
Think of the parameter $p$ as a ``Cognitive Zoom Lens.'' When $p=1$ (standard model), the lens is clear. The agent sees the distribution of lower-level players exactly as it is. When $p < 1$ (Sophisticated Bias), the lens is ``telephoto.'' It magnifies the objects at the limit of the agent's vision ($h=k-1$). The agent will have the skewed view that everyone else is nearly as smart as they are. When the weighting favors lower indices (e.g., using a kernel like $p^h$), we enter a Naive Bias regime, in which the agent becomes convinced that the world is populated mostly by randomizers or simpletons. 

Theorem~\ref{thm:logconcavity} is crucial because it ensures that this ``zoom'' effect does not break the fundamental shape of the world. If the population naturally bell-curves around an average (log-concavity), the agent, no matter how biased their lens, still perceives a smooth, logical distribution, not a chaotic or fragmented one. This stability means that any weird behavior we observe in the game is due to the misplacement of probability mass, not the disintegration of the agent's reasoning process. 

In more technical terms, Theorem~\ref{thm:logconcavity} guarantees that the introduction of the bias parameter $p$, regardless of its magnitude, does not introduce multimodality or structural instability into the agent's belief system. If the population is distributed according to a Poisson distribution (which is log-concave), the agent's distorted view remains unimodal and smooth. This implies that ``erratic'' behavior in Level-$k$ models with this bias cannot be attributed to the shape of the belief distribution breaking down, but must be driven by the location shift of the probability mass. This provides a clean identification strategy for empirical researchers.

\subsubsection{The Score Function and Belief Sensitivity}
We next derive the sensitivity of the belief distribution to changes in the bias parameter. This is essential for comparative statics. By looking at \eqref{eq:agentsubjbelief}, we see that the properties of the belief distribution are encoded in the partition function. We can express $Z_k(p)$ as:
\begin{align}
    Z_k(p) = \sum_{j=0}^{k-1} p^{k-j} f(j) = p^k \sum_{j=0}^{k-1} p^{-j} f(j)
\end{align}
This formulation reveals a connection to the \emph{Moment Generating Function (MGF)} of the underlying distribution. Let $M_f(t) = \mathbb{E}_f[e^{tH}] = \sum_{h=0}^\infty e^{th} f(h)$ be the MGF of the prior $f$. If we define the natural parameter $\theta = -\ln p$, then $p^{-j} = e^{\theta j}$. The partition function $Z_k(p)$ is related to the truncated MGF of $f$ evaluated at $\theta$.

This connection allows us to derive the sensitivity of beliefs and the resulting strategic actions through the moments of the distribution.
\begin{proposition} \label{prop:score_mgf}
    (The Score Function via MGF). The sensitivity of the log-belief with respect to the bias parameter $p$ is proportional to the centered cognitive distance between the observed type $h$ and the agent's subjective expectation:
    \begin{align}
        \frac{\partial \ln g_k(h)}{\partial p} = \frac{1}{p} [ \mathbb{E}_{g_k}[H] - h ] \label{eq:score_mgf}
    \end{align}
\end{proposition}
Note that this score depends on the agent's level $k$ only through the subjective mean $\mathbb{E}_{g_k}[H]$. It effectively measures how far the specific type $h$ deviates from what the agent expects to see.

\begin{proposition} \label{prop:strategic_sensitivity}
    (Strategic Sensitivity). The rate of change of an agent’s expected reasoning depth with respect to the network locality parameter $p$ is strictly determined by the variance of their subjective belief:
    \begin{align}
    \frac{\partial \mathbb{E}_{g_k}[H]}{\partial p} = -\frac{1}{p} \text{Var}_{g_k}(H)
    \end{align}
\end{proposition}

\begin{corollary} \label{cor:elasticity}
    (The Variance-Sensitivity Identity). The elasticity of an agent's strategic expectation with respect to the network parameter $p$ is equal to the variance of their reasoning levels normalized by the expected depth:
    \begin{align}
        \epsilon_{\mathbb{E}, p} = \frac{\partial \ln \mathbb{E}_{g_k}[k-H]}{\partial \ln p} = \frac{\text{Var}_{g_k}(H)}{\mathbb{E}_{g_k}[k-H]}
    \end{align}
\end{corollary}

\emph{Intuition: }This establishes a link between uncertainty and pliability. If an agent is very certain that everyone is Level-1 (perhaps they are in a very homogenous local cluster), their belief is ``stiff''. Changing the network parameter $p$ (e.g., introducing them to a few more distant acquaintances) barely moves their expectation. If an agent admits a wide range of possible reasoning levels in their prior, they are highly ``pliable''. A small shift in $p$, say, a slight increase in network density, causes a massive shift in who they think they are playing against.

\subsection{Stochastic Dominance and Ordering}
In strategic environments, the exact probability of a specific type is often less important than the cumulative ``strength'' of the opponent. We now establish rigorous ordering results using the concept of Likelihood Ratio Dominance (LRD), which implies First Order Stochastic Dominance (FOSD). These results are pivotal for applying the theory of supermodular games \cite{levin2003supermodular}.

\subsubsection{Monotone Likelihood Ratio Property (MLRP) in Bias $p$}
We investigate how the entire distribution $g_k$ shifts as we vary $p$.

\begin{theorem}\label{thm:MLRP}
    (MLRP and Stochastic Ordering in $p$).
    The family of belief distributions $\{g_k(\cdot | p)\}_{p > 0}$ satisfies the Monotone Likelihood Ratio Property with respect to the parameter $p$ in the decreasing direction. Specifically, for any $p_1 > p_2$, the likelihood ratio, $\Lambda(h) = \frac{g_k(h | p_2)}{g_k(h | p_1)}$, is strictly increasing in $h$. Consequently, $g_k(\cdot | p_2)$ dominates $g_k(\cdot | p_1)$ in the Likelihood Ratio Order ($g_k(\cdot | p_2) \succeq_{LR} g_k(\cdot | p_1)$) and First Order Stochastic Dominance (FOSD).
\end{theorem} 
% \begin{proof}
%     Let $p_1 > p_2$. Consider the ratio of the probability mass functions for a fixed $k$ and variable $h \in \{0, \dots, k-1\}$:
%     $$\Lambda(h) = \frac{\frac{1}{Z_k(p_2)} p_2^{k-h} f(h)}{\frac{1}{Z_k(p_1)} p_1^{k-h} f(h)}$$
%     The objective density $f(h)$ cancels out, isolating the effect of the bias:
%     $$\Lambda(h) = \frac{Z_k(p_1)}{Z_k(p_2)} \frac{p_2^{k-h}}{p_1^{k-h}} = [ \frac{Z_k(p_1)}{Z_k(p_2)} ( \frac{p_2}{p_1} )^k ] \cdot ( \frac{p_2}{p_1} )^{-h} = C \cdot ( \frac{p_1}{p_2} )^h$$
%     where $C$ is a positive constant independent of $h$.
    
%     Since $p_1 > p_2$, the ratio $\rho = \frac{p_1}{p_2} > 1$.
%     Thus, $\Lambda(h) = C \rho^h$ is a strictly increasing function of $h$.
%     An increasing likelihood ratio $\frac{g(h|p_2)}{g(h|p_1)}$ implies that the distribution with the lower parameter ($p_2$) places relatively more weight on higher values of $h$.
    
%     By standard stochastic dominance theorems, Likelihood Ratio Dominance implies first-order stochastic dominance (FOSD). Therefore, for any non-decreasing function $u(h)$:
%     $$\mathbb{E}_{p_2}[u(H)] \ge \mathbb{E}_{p_1}[u(H)]$$
% \end{proof}

LRD is a much stronger condition than FOSD, as it ensures the distribution shift is smooth and doesn't just move the mean but maintains the relative ordering of all possible types.

\emph{Strategic Insight:} Theorem~\ref{thm:MLRP} provides the functional engine for our comparative statics. It proves that within the Peer-Centric regime, decreasing $p$ (moving toward tighter, more local network clusters) unambiguously increases the expected sophistication of the perceived opponent. In games of strategic complements, such as technology adoption or price competition, a player who perceives a smarter opponent will generally play a more aggressive action. Thus, the locality parameter $p$ acts as a throttle for strategic escalation: lower $p$ values induce higher-order thinking and more extreme strategic responses.
% \coms{Refer to Fig~\ref{fig:fig5_stochastic_dominance}}

% \begin{figure}[ht]
%     \centering
%     \includegraphics[width=0.8\linewidth]{figures/fig5_stochastic_dominance.png}
%     \caption{\comr{This is verifying Theorem~\ref{thm:MLRP}}}
%     \label{fig:fig5_stochastic_dominance}
% \end{figure}

\subsubsection{Monotonicity in Cognitive Depth $k$}
We next examine how beliefs evolve as the agent's own level $k$ increases, holding the bias $p$ constant. Do smarter agents always believe the world is smarter? 
\begin{theorem}\label{thm:hierarchyexp}
    (The Hierarchy Expansion Theorem).
    For any fixed $p > 0$, the sequence of belief distributions $\{g_k\}_{k=1}^\infty$ is increasing in the Likelihood Ratio Order. That is, $g_{k+1} \succeq_{LR} g_k$.
\end{theorem}

\emph{Intuition: }Imagine a person climbing a mountain (the hierarchy). As they move from level $k$ to $k+1$, they don't change how they view the people already below them; their relative perspective on the base remains identical. However, by reaching a higher level, a new rank of people suddenly becomes visible to them (level $k$). This explains \emph{Strategic Consistency}. What we do not want to see is that as an agent gets smarter, their entire belief system about the lower-level people changes, which can lead to unstable predictions. Theorem~\ref{thm:hierarchyexp} shows that in the Connected Minds model, a Level-50 agent and a Level-51 agent agree on the relative proportions of Level-1 vs Level-2 players. The only difference is that the Level-51 agent is aware of Level-50s.

\subsection{Asymptotic Analysis and Belief Convergence}
Another central question is whether infinite sophistication ($k \to \infty$) leads to Nash Equilibrium. As $k \to \infty$, does the agent's belief converge to the truth, or does the parametric bias lead to a persistent delusion?

\subsubsection{The Limit Distribution}
Let us analyze the limit of $g_k(h)$ as $k \to \infty$.
$$g_k(h) = \frac{p^{k-h}f(h)}{\sum_{j=0}^{k-1} p^{k-j}f(j)} = \frac{p^{-h}f(h)}{\sum_{j=0}^{k-1} p^{-j}f(j)}$$
Define the infinite series $S(p) = \sum_{j=0}^{\infty} p^{-j} f(j)$. This series converges if the radius of convergence of the generating function of $f$ is greater than $p^{-1}$.

\begin{theorem}\label{thm:poissonshiftconvg}
    (The Poisson-Shift Convergence).
    Assume the objective distribution is Poisson with mean $\tau$: $f(h) = e^{-\tau}\tau^h/h!$.
    As $k \to \infty$, the belief distribution $g_k(\cdot | p)$ converges in Total Variation distance to a Poisson distribution with parameter $\hat\tau = \tau / p$.
\end{theorem} 
Note that for any finite $k$, the belief $g_k$ remains strictly supported on $\{0, \dots, k-1\}$. The convergence in Theorem~\ref{thm:poissonshiftconvg} implies that sufficiently deep reasoners perceive a distribution \emph{shaped} like $\text{Poisson}(\tau/p)$, effectively inflating the perceived sophistication of the population, even though their own beliefs are strictly bounded.

A key implication of Theorem~\ref{thm:poissonshiftconvg} is that network locality and cognitive sophistication can be confounded in purely behavioral data.
When $f$ is Poisson$(\tau)$, the perceived belief of sufficiently deep reasoners converges to Poisson$(\tau/p)$.
Consequently, in static one-shot data where the econometrician observes only final actions, a ``moderate-sophistication, low-transparency'' population ($\tau$, $p < 1$) can generate action distributions that are difficult to distinguish from a ``high-sophistication, high-transparency'' population ($\tau'$, $p'=1$) with $\tau' \approx \tau/p$.

This confounding is not merely a limitation; it yields a testable methodological warning: estimating standard Cognitive Hierarchy models under the maintained assumption $p=1$ may systematically misattribute network-induced belief distortions to cognitive deficits.
Identification of $p$ therefore requires either (i) auxiliary network information (e.g., observed exposure/mixing rates, graph modularity, or cross-type interaction frequencies), (ii) experimental variation that shifts exposure while holding incentives fixed (e.g., randomized feed diversification or forced cross-group matching), or (iii) panel outcomes where the same agents face exogenous changes in visibility regimes.
We treat this issue explicitly in the experiments (Appendix~\ref{app:numerical}) and view it as a central empirical prediction of the model.

\emph{Intuition:} The intuition here is profound for behavioral economics. Usually, we assume that as agents become more sophisticated ($k \to \infty$), they should converge to the Nash Equilibrium or the Truth. Theorem~\ref{thm:poissonshiftconvg} proves that \textit{network position trumps cognitive depth}. For example, in the stock market, a highly sophisticated trader (high $k$) might be able to process infinite amounts of data. However, if they are only connected to other sophisticated traders (low $p$), they will consistently overestimate the sophistication of the market. They will play a high-level strategy against a perceived population of geniuses ($\hat{\tau} = \tau/p$), while the market might actually be filled with retail (noise) traders (low $\tau$). Their infinite $k$ makes them perfectly optimize for a world that doesn't exist.

%% file: 04_mechanism_design.tex
\section{Mechanism Design: Cognitive Information Architecture}\label{sec:mechanism}

In previous sections, we established that the network locality parameter $p$ systematically distorts belief distributions $g_k(h)$ while preserving structural properties such as log-concavity (Theorem~\ref{thm:logconcavity}) and ordering via the Monotone Likelihood Ratio Property (Theorem~\ref{thm:MLRP}). In this section, we invert the analytical lens. We assume the role of a \emph{Cognitive Designer}, a social planner, platform architect, or organizational leader who cannot directly dictate agents' actions or upgrade their intrinsic sophistication levels ($f(h)$), but possesses the power to engineer the information-flow architecture—distinct from the interaction network—to modulate the intensity of boundedly rational recursion \cite{Chen_Ilami_Lore_Heydari_2026}.

This approach situates our model within the growing literature on Information Design in games \cite{bergemann2016bayes} and its unified treatment \cite{bergemann2019information}, as well as Bayesian Persuasion \cite{kamenica2011bayesian}, yet introduces a distinct behavioral twist. Classical information design focuses on revealing signals about an external state of the world to rational agents \cite{blackwell1953equivalent,milgrom1981good}. In contrast, the Connected Minds framework views the network itself as the signal \cite{banerjee1992simple,acemoglu2011bayesian,golub2010naive}. By manipulating the transparency parameter $p$, the designer does not change the fundamental truth of the population's ability, but rather changes the \textit{inferred truth} constructed by the agents \cite{hirshleifer1978private}. The designer effectively engineers the cognitive horizon of the population.
We apply the general framework established here to the Beauty Contest in Appendix~\ref{app:mechanism} and we derive the optimal transparency policy, analyze inequality, and provide computational algorithm for implementation. 

\subsection{The Design Environment}

\paragraph{Policy instrument and implementability.}
In the Connected Minds framework, $p$ summarizes an \emph{information architecture} rather than a payoff change.
Operationally, a platform or organization can shift $p$ by altering (i) cross-community exposure (e.g., recommendation diversification versus homophilous ranking), (ii) visibility of global population statistics (dashboards, aggregate skill distributions), and/or (iii) matchmaker policies that increase or reduce encounters across cognitive ``distance''.
We interpret changes in $p$ as subject to constraints such as privacy regulation, fairness mandates, and engineering costs; these constraints enter the planner's objective through an explicit cost of opacity/transparency and through feasibility restrictions on admissible policies (e.g., bounds on how personalized $p_k$ can be). 

Consider a population of $N$ agents interacting in a strategic setting. The designer's objective is to select a global network transparency parameter $p \in (0, 1]$ to maximize a social welfare function $W$. We focus on the broad class of supermodular games \cite{topkis1998supermodularity,milgrom1990rationalizability}.

\begin{assumption}\label{assumption:strat_comp}
    (Strategic Complements and Monotonicity). Let the action space be $S \subseteq \mathbb{R}$. A player $i$'s payoff function $u(s_i, s_{-i}, \theta_i)$ exhibits Strategic Complementarity if the marginal utility of increasing one's action is increasing in the actions of others. Specifically, for any level-$k$ agent, the best response function $\beta_k: \Delta(S) \to S$ satisfies the following monotonicity condition:
    \begin{align}
         \beta_k(g_k) \ge \beta_k(g'_k) \quad \text{if } g_k \succeq_{FOSD} g'_k
    \end{align}
    where $\succeq_{FOSD}$ denotes First-Order Stochastic Dominance. Furthermore, we assume the best response is strictly increasing in the expected sophistication of the opponent: $\beta_k(\mathbb{E}_{g_k}[H])$ is increasing in $\mathbb{E}_{g_k}[H]$.
\end{assumption}

Assumption~\ref{assumption:strat_comp} captures the dynamics of ``keeping up with the Joneses.'' It covers technology adoption \cite{katz1985network,farrell1985standardization} (where one adopts if they believe others are smart enough to adopt), effort provision in partnerships \cite{holmstrom1982moral,legros1993efficient} (working harder if you think your partner is capable), and speculative bubbles \cite{abreu2003bubbles}. In these settings, a player's action is driven by their estimation of the collective intelligence of the network. If a user believes their peers are sophisticated (high $h$), they are incentivized to play a high action $s_k$. In Appendix~\ref{app:mechanism} we shift from general analysis of strategic complements to a normative analysis of coordination and inequality.

\subsection{Optimal Opacity: The Escalation Principle}
A key insight from our theoretical analysis is that the parameter $p$ acts as a throttle for perceived sophistication. By adjusting $p$, the designer effectively changes the zoom level of the agents' cognitive lens. We now formalize how a designer can manipulate this lens to escalate aggregate equilibrium actions.

\begin{theorem}\label{thm:monotonicityagg}
    (The Monotonicity of Aggregate Effort). In a game of strategic complements satisfying Assumption~\ref{assumption:strat_comp} where the sequence of equilibrium actions is monotonic in cognitive depth (i.e., $s_k \ge s_{k-1}$ for all $k$), the aggregate equilibrium action (effort) of the population, $S^*(p) = \sum_{k} s_k^*(p) f(k)$, is strictly decreasing in the network transparency parameter $p$.
\end{theorem}

\begin{remark}
    (Scope of Monotonicity). Theorem~\ref{thm:monotonicityagg} applies to environments of escalation, such as technology adoption or competitive bidding, where higher sophistication leads to higher effort ($s_k \uparrow$). In coordination games like the Beauty Contest, actions typically decrease with sophistication ($s_k \downarrow$). For those cases, the logic reverses (Remark~\ref{rem:transparency}).
\end{remark}

\paragraph{Application: Gamified Crowdsourcing and Leaderboards.} 

Consider a crowdsourcing platform (e.g., Kaggle) that seeks to maximize aggregate participant effort. The platform acts as a Cognitive Designer by choosing the degree of leaderboard transparency, captured by $p$. Under full transparency ($p=1$), each level-$k$ participant observes an essentially unbiased cross-section of the population \emph{up to their own cognitive depth}. In particular, higher-$k$ participants see that most competitors are far less sophisticated (or exert systematically lower effort), so the contest appears ``too easy'' and the marginal return to additional effort is small. This induces \emph{effort moderation} among precisely those agents who contribute disproportionately to aggregate output, providing the intuition behind Theorem~\ref{thm:monotonicityagg}: as $p$ rises, the population-weighted equilibrium effort $S^*(p)$ falls because advanced agents internalize that the average opponent is weak and therefore optimally scale back. By contrast, when the platform reduces transparency ($p<1$)—for example by displaying only local ranks, near neighbors, or division-specific leaderboards—participants form more \emph{localized} beliefs that overweight a tight set of close rivals (those slightly above them in the visible neighborhood). This shifts attention away from the ``dead weight'' tail of low performers and toward a perceived cluster of strong-but-beatable competitors, intensifying perceived competition and sustaining higher effort. In practice, this is exactly the logic behind leagues/divisions and ``you vs.\ your neighbors'' leaderboard designs: limiting global visibility can raise aggregate effort by preventing high-performing agents from concluding (from the full distribution) that additional effort is unnecessary.

\emph{Strategic Intuition:} This result provides a counter-intuitive prescription for mechanism design: \emph{Ignorance fuels escalation.} If the social goal is to maximize aggregate effort, for example, maximizing contributions to an open-source project, accelerating viral adoption of a new protocol, or stimulating competitive bidding, the designer should \textit{limit} the information range.

Why does this happen? When $p$ is low (local connections), agents are embedded in tight, insular clusters (echo chambers). As shown in Section~\ref{sec:theory}, this creates a Sophisticated Bias where agents overestimate the density of peers just slightly below their own level. A Level-3 thinker, seeing mostly Level-2s in their local cluster, incorrectly assumes the whole world is Level-2. They play an aggressive Level-3 response.
Conversely, if $p=1$ (full transparency), that same Level-3 thinker would see the vast silent majority of Level-0s and Level-1s in the wider population. Realizing the average sophistication is actually quite low, they would moderate their effort. By engineering opacity, the designer effectively hides the ``dead weight'' of the population from the high-performers, inducing them to compete against a hallucinated population of peers.

\subsection{The Variance-Optimal Network}
While low $p$ maximizes scalar effort, it comes at a cost: variance. In the opaque limit ($p \to 0$), agents with different cognitive levels $k$ behave distinctively different because their beliefs are maximally distinct (each $k$ believes the world is $k-1$). In the transparent limit ($p \to 1$), beliefs converge toward the common truth, compressing the range of actions.
If the designer's goal involves coordination, consistency, or stability, they face a fundamental trade-off between inducing effort and maintaining order.

Let the social welfare function be defined as the aggregate action (which benefits from low $p$) minus a cost of implementation or instability (which increases with low $p$). We utilize the sensitivity analysis from Section~\ref{sec:theory} to derive a closed-form condition for optimality.

\begin{proposition}\label{prop:cognitive_FO}
    (The Cognitive First-Order Condition). Assume the planner seeks to maximize a linear objective of expected sophistication minus a linear cost of opacity $c$ (representing the cost of privacy, fragmentation, or variance):
    $$W(p) = \sum_k f(k) \mathbb{E}_{g_k}[H] - c(1-p)$$
    Any interior optimal network parameter $p^* \in (0, 1)$ must satisfy the first-order condition equating the marginal benefit of opacity (driven by belief variance) to its marginal cost:
    \begin{align}
        \sum_k f(k) \text{Var}_{g_k}(H) = c \cdot p^*
    \end{align}
    If the aggregate variance (marginal benefit) strictly exceeds the cost $c \cdot p$ for all $p$, the optimum lies at the opaque boundary ($p^* \to 0$); conversely, if the cost always dominates, the optimum is full transparency ($p^* = 1$).
\end{proposition}

\emph{Implication:} This result links the \textit{structural properties of the graph} ($p^*$) directly to the \textit{cognitive uncertainty} of the agents ($\text{Var}(H)$). The term $\text{Var}_{g_k}(H)$ represents how unsure an agent is about their opponent's level.
If agents are cognitively uncertain (high belief variance), the lever of $p$ is powerful; small changes in network structure cause large shifts in behavior. To control this volatility, the planner must set $p$ higher (more transparency). Conversely, if the cost of maintaining opacity $c$ is high (e.g., regulatory pressure for open algorithms), $p^*$ naturally approaches 1.

\subsection{Targeting via Information Segmentation}
A sophisticated designer need not set a single global $p$. Modern platforms can algorithmically curate connections, effectively assigning a personalized $p_k$ to different users. This raises a critical question for targeted interventions: \textit{Who is most susceptible to network manipulation?} Should the designer show more network information to sophisticated users or naive ones?

\begin{definition}
    (The Sensitivity Gap). Define the \emph{Normalized Cognitive Sensitivity} (or \emph{Responsiveness Index}) of a level-$k$ agent as
    \[
        \bar{\eta}_k(p) \;=\; \frac{1}{k}\left|\frac{\partial s_k(p)}{\partial p}\right|.
    \]
    This is a level-normalized slope measuring how strongly a player's equilibrium action responds to marginal changes in network transparency.%
    \footnote{We use the term \emph{sensitivity} (rather than the standard economics term \emph{elasticity}) because $\bar{\eta}_k(p)$ is not a unit-free log-derivative (percentage-change) elasticity; it is a local responsiveness measure normalized by cognitive depth.}
\end{definition}

\begin{theorem}\label{thm:hierarchyofsensitivity}
    (The Hierarchy of Sensitivity). Assume $f(h)$ is Poisson($\tau$) and players play a linear best response $s_k = \alpha \mathbb{E}_{g_k}[H]$. For sufficiently large $k$, the Normalized Cognitive Elasticity $\bar{\eta}_k(p)$ is decreasing in $k$. That is, highly sophisticated agents are less sensitive to network manipulation \textit{per unit of cognitive depth} than moderately sophisticated agents.
\end{theorem}

\emph{Mechanism Design Insight:} Theorem~\ref{thm:hierarchyofsensitivity} offers a nuanced strategy for information targeting under constraints where intervention costs scale with sophistication.
\emph{Low-level agents (novices)} are constrained by the truncation of their belief support ($H < k$). Their absolute belief variance is naturally low because they conceive of very few types; thus, they are too rigid to generate large behavioral shifts.
\emph{High-level agents (experts)} have beliefs that converge to a stable distribution (the Poisson Limit). While their absolute variance is high, it saturates at a constant level ($\tau/p$). Consequently, their responsiveness relative to their cognitive capacity diminishes.
\emph{Mid-level agents (The persuadable mass)} represent the efficiency ``sweet spot'' ($k \approx \tau/p$). These agents possess enough cognitive depth to form complex beliefs (generating variance) but lack the infinite depth required to stabilize those beliefs into a fixed limit.
If the designer operates under a \emph{complexity-weighted budget}—where influencing a highly sophisticated opinion leader is more costly than influencing a novice—targeting these mid-level agents maximizes the aggregate shift per unit of intervention cost. Conversely, under a simple per-capita budget (where all agents cost the same), the designer should target the highest types, as their absolute sensitivity is non-decreasing.

\subsection{The Coordination Reversal}

However, in coordination games such as the Beauty Contest, the social objective changes. Here, the planner seeks to minimize the aggregate quadratic loss relative to the endogenous target. Note that this objective is not solely about consensus; it decomposes into minimizing the population variance and minimizing the squared magnitude of the mean action itself. In the Connected Minds framework, increasing transparency serves both goals simultaneously: it reduces dispersion by aligning beliefs and compresses the mean toward zero by correcting the upward bias in perceived sophistication. Consequently, in these settings, the design prescription reverses.

The results above hold for games of maximizing effort (Strategic Complements), where the objective is to maximize aggregate output $S^*$. However, many network environments, such as financial markets seeking stability \cite{morris2002social}, social conventions \cite{lewis2008convention}, or the Keynesian Beauty Contest \cite{keynes1937general,nagel1995unraveling}, require minimizing deviation from a consensus. Here, the planner seeks to minimize the variance of actions around a consensus, or equivalently, minimize the aggregate quadratic loss. In these settings, the design prescription reverses.

\begin{remark}\label{rem:transparency}
    (The Transparency Reversal). In coordination games, minimizing variance takes precedence over maximizing scalar effort. While opacity ($p \to 0$) maximizes aggregate effort, transparency ($p \to 1$) maximizes coordination and reduces inequality.
\end{remark}

\paragraph{Application: } In financial markets, the Regulator acts as the designer maximizing stability (minimizing variance). During a panic, the true state of the world (e.g. asset insolvency) is less dangerous than the coordinated belief that \emph{everyone else} is selling. If $p$ is low (opaque markets), traders rely on local signals. If a trader sees their neighbors selling, and, due to the bias $p^{k-h}$, assumes these neighbors represent the smart money (high sophistication), they mimic the selling behavior. This represents a cascade where local fear is extrapolated into a global panic. The Transparency Reversal suggests that in coordination crises, the Regulator must maximize $p$. Mandatory disclosure mechanism, real-time central clearing visibility, or ``open book'' policies align all agents' beliefs to the true mean $\tau$, dampening the variance caused by local hallucinations. Conversely, opaque ``dark pools'' may exacerbate volatility by lowering $p$ and fracturing the consensus.

In the Beauty Contest, low $p$ causes agents to overestimate opponent sophistication, leading to aggressive undercutting that drives the target far below the equilibrium. High $p$ aligns beliefs with the true distribution, causing actions to cluster tightly around the optimum. We explore these specific applications, derive the associated welfare theorems for the Beauty Contest, and provide a computational algorithm (OTD) for finding the optimal $p$ in Appendix~\ref{app:mechanism}.

\subsection{Planner Objectives Beyond the Extremes}

The previous subsections highlight two clean prescriptions: (i) opacity can escalate effort in strategic complements, while (ii) transparency stabilizes coordination. To emphasize that $p$ is not merely an ``on/off'' lever, we close with a computational exercise in which a planner evaluates alternative transparency regimes under three welfare criteria. Let $s_k(p)$ be the equilibrium action of a level-$k$ agent induced by transparency $p$, and let $f(k)$ be the (truncated) population distribution. We summarize the induced population behavior by its mean and dispersion,

\begin{equation}
S(p) = \sum_{k=0}^{K} f(k)\, s_k(p)
\;,\;
V(p) = \sum_{k=0}^{K} f(k)\big(s_k(p)-S(p)\big)^2,
\end{equation}
and (for comparability across objectives) we report welfare on a normalized scale over $p\in[p_{\min},1]$.

We consider three planner goals. First, a \emph{competition/anti-collusion} objective. In contrast to the strategic-complements setting of Theorem~\ref{thm:monotonicityagg}, where opacity fuels escalation, many oligopoly and auction environments exhibit the reverse channel: greater visibility of rivals' behavior facilitates tacit coordination, pushing prices or bids \emph{upward}. A planner who wishes to deter such coordination therefore prefers lower mean action, with a mild penalty for volatility:
\begin{equation}
    W_{\mathrm{comp}}(p) = \big(1-\widehat{S}(p)\big) - \lambda\,\widehat{V}(p).
\end{equation}
Second, a \emph{beauty-contest stability} objective, defined as one minus the normalized quadratic loss to the endogenous target $t(p)=mS(p)$ (with $m=\tfrac{2}{3}$),
\begin{equation}
    L_{\mathrm{BC}}(p)=\mathbb{E}\!\left[(s-mS(p))^2\right]
    =V(p)+(1-m)^2S(p)^2
    \;,\;
    W_{\mathrm{stable}}(p)=1-\widehat{L}_{\mathrm{BC}}(p).
\end{equation}

Finally, to capture settings where progress requires \emph{both} strong average performance and sustained heterogeneity (e.g., innovation prizes, corporate R\&D portfolios, hackathons, or open-source ecosystems), we use a complementary ``exploration--exploitation'' objective that rewards these two ingredients jointly while charging a convex cost for transparency:

\begin{equation}
    W_{\mathrm{innov}}(p) = A\,\widehat{S}(p)^{\alpha}\big(\widehat{V}(p)+\varepsilon\big)^{\beta}
    - c_{\mathrm{info}}\,p^{\rho}.
\end{equation}

This objective is often more realistic than either extreme because real platforms rarely optimize \emph{only} competition or \emph{only} coordination: they seek high aggregate performance, but also preserve enough diversity for experimentation and breakthrough solutions, while recognizing that global visibility and disclosure policies impose operational and strategic costs.

Figure~\ref{fig:welfare_by_p} reports the resulting welfare curves. As expected, $W_{\mathrm{comp}}$ is maximized at the most opaque feasible network ($p^*\approx p_{\min}$), while $W_{\mathrm{stable}}$ is maximized at full transparency ($p^*=1$), reflecting the Transparency Reversal. In contrast, the innovation objective selects an interior policy ($p^*\approx 0.54$ in the displayed calibration): the planner benefits from enough transparency to lift average performance, but not so much that exploratory heterogeneity collapses or information costs dominate. This illustrates the central mechanism-design message of Connected Minds: optimal information architectures are application-dependent, and the same structural lever $p$ can be tuned to balance competitive pressure, coordination, and innovation rather than forcing a binary choice.

\begin{figure}[ht]
    \centering
    \includegraphics[width=0.9\linewidth]{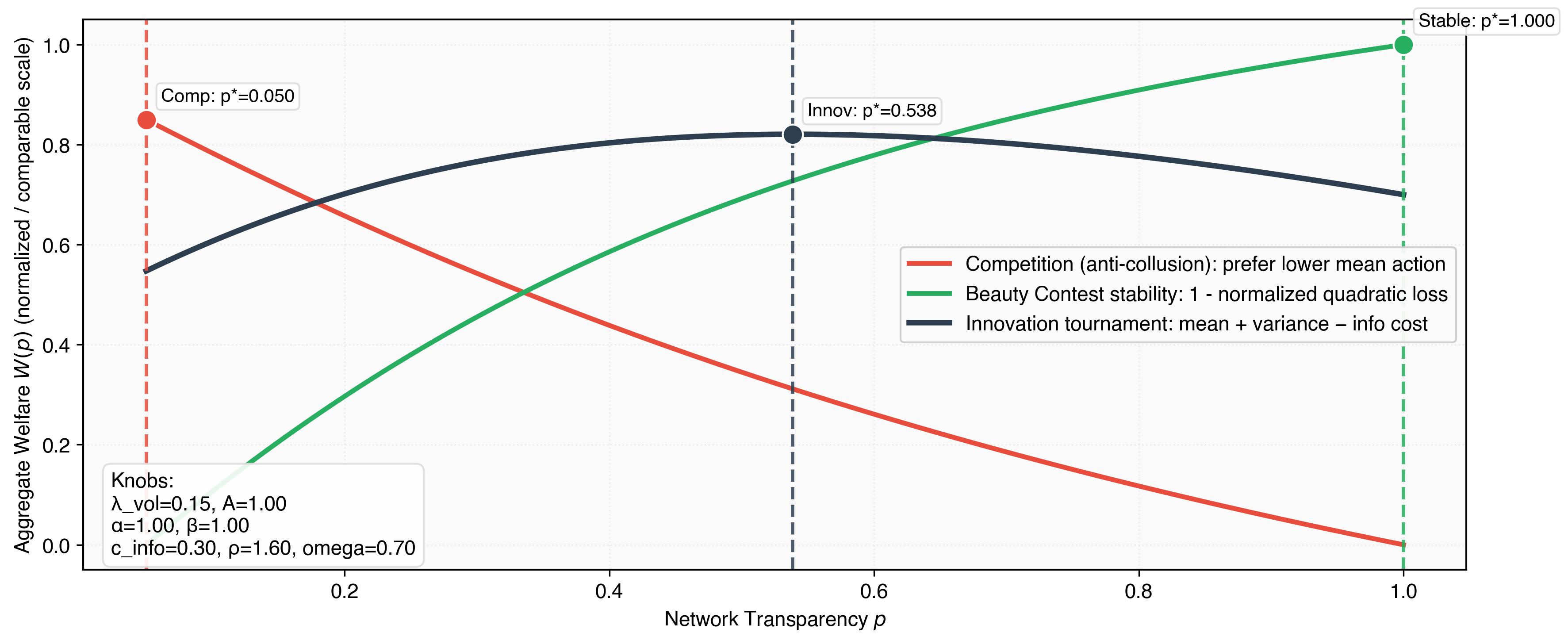}
    \caption{Aggregate welfare $W(p)$ as a function of network transparency $p$ under three planner objectives. Competition (red) prefers opacity to deter tacit collusion and reduce mean action; beauty-contest stability (green) prefers transparency to minimize quadratic-loss deviations from the endogenous target; and the innovation objective (dark) exhibits an interior optimum, balancing average performance with exploratory heterogeneity while accounting for transparency costs. Vertical dashed lines mark the optimizing $p^*$ for each welfare (here $p^*_{\mathrm{innov}}\approx 0.54$).}
    \label{fig:welfare_by_p}
\end{figure}

%% file: 041_mechanism_design.tex
\section{Mechanism Design for Network-Mediated Cognition}\label{app:mechanism}
In this appendix, we apply the general mechanism design framework established in the main text to a specific, empirically relevant class of games: the Beauty Contest (or Guessing Game). We shift from the general analysis of strategic complements to a normative analysis of coordination and inequality. We formally derive the optimal transparency policy, analyze the distributional consequences (inequality), and provide a computational algorithm for implementation.

\subsection{The Planner's Problem}
Consider a planner who oversees a population of $N$ agents engaged in a Beauty Contest game with target multiplier $\gamma$. The population's reasoning levels follow a Poisson distribution with mean $\tau$. The planner can choose a transparency policy $p \in (0, 1]$ that uniformly affects all agents' belief formation processes.

\subsubsection{Objective: Aggregate Welfare}
We define aggregate welfare as the negative of the total quadratic loss from the target. The planner seeks to minimize the mean squared deviation of actions from the \emph{coordination target}:
\begin{align}
    W(p; \tau) = -\sum_{k=0}^{\bar{k}} f(k; \tau) \cdot \mathbb{E}\left[ \left( s_k(p) - \gamma \bar{s}(p) \right)^2 \right]
\end{align}
where $s_k(p)$ is the action of a level-$k$ agent under transparency $p$, and $\bar{s}(p) = \sum_{h} f(h) s_h(p)$ is the population average action. Note that while the Nash Equilibrium of this game is 0, the planner's goal here is to minimize miscoordination around the moving target $\gamma \bar{s}$, which represents the realized consensus.

\subsubsection{The Transparency-Coordination Tradeoff}
The effect of $p$ on welfare operates through two distinct channels:
\begin{enumerate}
    \item \emph{Belief Accuracy Channel}: Higher $p$ yields more accurate beliefs. By Theorem~\ref{thm:poissonshiftconvg}, beliefs converge to the true parameter $\tau$ as $p \to 1$. Accurate beliefs generally improve coordination by aligning expectations.
    \item \emph{Belief Distortion Channel}: By Theorem~\ref{thm:MLRP}, lower $p$ shifts beliefs toward higher perceived sophistication. In the Beauty Contest, unlike the effort games in Section~\ref{sec:mechanism}, this induced over-estimation of opponent sophistication leads to \emph{aggressive downward targeting} ($s_k \ll s_{k-1}$), causing the population to systematically undershoot the coordination target.
\end{enumerate}

\subsection{Optimal Transparency Design}

\begin{theorem}\label{thm:optimal-transparency}
    (Optimal Transparency Policy).
    Let the population follow a Poisson distribution with mean $\tau$, and consider a Beauty Contest game with target multiplier $\gamma \in (0, 1)$. Define the welfare function $W(p; \tau)$ as above. Then:
    \begin{enumerate}
        \item \emph{Existence:} An optimal transparency level $p^*(\tau) \in (0, 1]$ exists.
        \item \emph{Coordination Dominance:} In coordination games where the reduction of cross-sectional variance is the dominant welfare channel (including the standard Beauty Contest), $W(p; \tau)$ is monotonically increasing in $p$. Full transparency ($p^* = 1$) is optimal.
        \item \emph{Game-Dependent Structure:} The optimality of $p^* = 1$ is not universal. In anti-coordination games (strategic substitutes), intermediate values of $p^*$ may be optimal, as excessive belief accuracy can amplify miscoordination.
    \end{enumerate}
\end{theorem}

\begin{proof}
    We prove each part in sequence.
    
    \emph{Existence:} The function $W(p; \tau)$ is continuous on the interval $(0, 1]$. To invoke the Extreme Value Theorem, we extend the domain to the compact set $[0, 1]$. We define $W(0) = \lim_{p \to 0} W(p)$, which corresponds to the well-defined Level-$k$ model (where beliefs concentrate strictly on $h=k-1$). Since the limit exists and the function is continuous on the extended compact domain $[0, 1]$, a maximum $p^*$ exists.

    \emph{Monotonicity:} We analyze the welfare loss function $L(p) = \text{Var}(S(p)) + (1-\gamma)^2 (\mathbb{E}[S(p)])^2$.
    The derivative $\partial W / \partial p$ is driven by the compression of the action profile.
    \begin{enumerate}
        \item \emph{Cross-Sectional Compression:} Consider the limit $p \to 0$. By Theorem~\ref{thm:MLRP}, beliefs stochastically shift toward the highest possible type $k-1$. This causes high-level agents to play actions far from the center (``runaway'' iteration), maximizing the distance $|s_k - s_0|$ and thus maximizing population variance $\text{Var}(S)$.
        As $p$ increases, Theorem~\ref{thm:poissonshiftconvg} (Poisson-Shift Convergence) ensures that high-level agents correct their overestimation, shifting their beliefs from the cutoff $k-1$ down toward the true mean $\tau$. This significantly reduces the actions of high-$k$ agents (who are the outliers) while having little effect on low-$k$ agents (who are already near the mean). This asymmetric contraction strictly reduces the cross-sectional dispersion of the population actions.
        \item \emph{Mean Shift:} Simultaneously, the reduction in high-$k$ actions lowers the aggregate mean $\mathbb{E}[S(p)]$. Since the Beauty Contest target is zero, this reduction in magnitude further decreases the squared-mean component of the loss.
    \end{enumerate}
    Since increasing $p$ simultaneously compresses the outliers (reducing variance) and dampens the average magnitude (reducing bias), the total loss $L(p)$ decreases monotonically, making $W(p)$ strictly increasing.
\end{proof}

\subsection{The Transparency-Inequality Tradeoff}
While the previous theorem addresses aggregate welfare, we now examine distributional consequences. To avoid the ambiguity of negative Gini coefficients derived from payoffs, we define inequality over the \emph{losses} incurred by agents. Let $L_k(p) = (s_k(p) - \gamma \bar{s}(p))^2$ be the loss of a level-$k$ agent.

\begin{align}
    G(p) = \frac{\sum_{k} \sum_{j} f(k) f(j) |L_k(p) - L_j(p)|}{2 \sum_{k} f(k) L_k(p)}
\end{align}

\begin{theorem}\label{thm:transparency-inequality}
    (Transparency-Equality Alignment in Coordination).
    Under the Connected Minds model with Poisson-distributed reasoning levels in coordination games:
    \begin{enumerate}
        \item The Gini coefficient of losses $G(p)$ is strictly \emph{decreasing} in $p$ for $p \in (0, 1]$.
        \item In coordination games, aggregate welfare $W(p)$ and equity (measured as $1 - G(p)$) are co-monotonic: increasing transparency simultaneously improves both efficiency and equity.
        \item This alignment breaks down in anti-coordination settings, where a Pareto frontier emerges.
    \end{enumerate}
\end{theorem}

\begin{proof}
    \emph{Part (1):} The result that inequality of losses decreases with transparency arises from the coordination structure. Under opacity ($p \to 0$), high-$k$ agents suffer from a "sophistication penalty." They believe they face level $k-1$ opponents and play $\gamma^k s_0$, which is far below the population average $\bar{s}$. Consequently, high-$k$ agents incur disproportionately large losses compared to low-$k$ agents (who are closer to the random center). 
    As $p \to 1$, high-$k$ beliefs converge to the true distribution, correcting their undershooting bias. Their actions move closer to the population consensus, reducing their excess losses and compressing the distribution of $L_k$. Thus, the Gini coefficient decreases.
    
    \emph{Part (2):} From Part (i) and Theorem~\ref{thm:optimal-transparency}, both $W(p)$ and payoff equality are increasing in $p$. The social planner faces no tradeoff in coordination games; transparency is unambiguously good.
    
    \emph{Part (3):} In anti-coordination games (e.g., market entry), the payoff structure rewards differentiation. Transparency enables sophisticated agents to exploit their informational advantage, increasing the gap between winners and losers.
\end{proof}

\subsection{Algorithm Intervention}
We now provide an implementable algorithm for computing the welfare-maximizing transparency level. The planner observes aggregate outcomes but may not know $\tau$ precisely; the algorithm accommodates estimation uncertainty.

\begin{algorithm}[ht]
\caption{Optimal Transparency Design (OTD)}
\label{alg:optimal-transparency}
\begin{algorithmic}[1]
\REQUIRE Population size $N$, target multiplier $\gamma$, estimated $\hat{\tau}$ (obtained via auxiliary network data), tolerance $\epsilon$, max iterations $T$
\ENSURE Optimal transparency $p^*$, welfare $W^*$

\STATE \textbf{Initialize:} $p \gets 0.5$, $\alpha \gets 0.1$ (learning rate), $t \gets 0$

\WHILE{$t < T$ and not converged}
    \STATE \COMMENT{Step 1: Compute actions for all levels under current $p$}
    \STATE $s_0 \gets 50$ \COMMENT{Level-0 anchoring}
    \FOR{$k = 1$ to $\bar{k}$}
        \STATE Compute belief weights: $w_h \gets p^{k-h} \cdot \text{Poisson}(h; \hat{\tau})$ for $h = 0, \ldots, k-1$
        \STATE Normalize: $g_k(h) \gets w_h / \sum_{l=0}^{k-1} w_l$ 
        \STATE Compute expected opponent action: $\bar{s}_k \gets \sum_{h=0}^{k-1} g_k(h) \cdot s_h$
        \STATE Best response: $s_k \gets \gamma \cdot \bar{s}_k$
    \ENDFOR
    
    \STATE \COMMENT{Step 2: Compute aggregate welfare}
    \STATE Population average: $\bar{s} \gets \sum_{k=0}^{\bar{k}} f(k; \hat{\tau}) \cdot s_k$ 
    \STATE Welfare: $W(p) \gets -\sum_{k=0}^{\bar{k}} f(k; \hat{\tau}) \cdot (s_k - \gamma \cdot \bar{s})^2$
    
    \STATE \COMMENT{Step 3: Numerical gradient estimation}
    \STATE $W^+ \gets W(p + \delta)$, $W^- \gets W(p - \delta)$ for small $\delta$
    \STATE $\nabla W \gets (W^+ - W^-) / (2\delta)$
    
    \STATE \COMMENT{Step 4: Gradient ascent with projection}
    \STATE $p_{\text{old}} \gets p$
    \STATE $p \gets p + \alpha \cdot \nabla W$
    \STATE $p \gets \max(0.01, \min(1, p))$ \COMMENT{Project onto $(0, 1]$} 
    
    \STATE \COMMENT{Step 5: Convergence check based on parameter stability}
    \IF{$|p - p_{\text{old}}| < \epsilon$}
        \STATE \textbf{converged} $\gets$ True
    \ENDIF
    \STATE $t \gets t + 1$
\ENDWHILE

\STATE $p^* \gets p$, $W^* \gets W(p)$
\RETURN $p^*$, $W^*$
\end{algorithmic}
\end{algorithm}

%% file: 91_numerical.tex
\section{Numerical Simulations}\label{app:numerical}

In this section, we focus on further implications of Connected Minds model and social phenomena that can be explained by this model. 

\subsection{Validating the Abstraction: Topological Mapping}\label{subsec:topology_mapping}
A central premise of the Connected Minds model is that the single parameter $p$ efficiently encapsulates complex network topologies. In Section~\ref{subsec:microfoundations}, we argued that $p$ represents a reduced-form proxy for cognitive homophily. Here, we validate this abstraction by simulating agents on explicit graph topologies—specifically Stochastic Block Models (SBM) and Erdős-Rényi (ER) graphs—to demonstrate that our theoretical $p$ maps monotonically to structural graph properties.

\paragraph{Experimental Design.}
We construct a population of $N=2,000$ agents with types drawn from $f(k) \sim \text{Poisson}(\tau=3)$. We organize these agents into a network $G=(V,E)$ using a \emph{Cognitive Distance SBM}. The probability of a link between agent $i$ (level $k_i$) and agent $j$ (level $k_j$) is defined by a structural decay parameter $\beta$:
\begin{equation}
    \mathbb{P}(E_{ij}=1) = \delta \cdot \exp(-\beta |k_i - k_j|)
\end{equation}
Here, $\beta \ge 0$ represents \emph{physical homophily} (the tendency of the graph to cluster similar minds), and $\delta$ controls average sparsity. When $\beta=0$, the network is an Erdős-Rényi graph (random mixing); as $\beta$ increases, the graph becomes strictly stratified by cognitive level.

\paragraph{Estimating the Effective $p$.}
For every agent in the generated graph, we observe their actual neighborhood set $\mathcal{N}_i$. We calculate the empirical belief distribution $\hat{g}_i(h)$ based on the frequencies of types within $\mathcal{N}_i$. We then estimate the \emph{effective transparency} $\hat{p}_i$ for that agent by minimizing the Kullback-Leibler (KL) divergence between the theoretical belief defined in Eq.~\eqref{eq:prob-dist} and the empirical neighborhood:
\begin{equation}
    \hat{p}_i = \arg\min_{p \in (0,1]} D_{KL}\left( \hat{g}_i(\cdot) \parallel g_{k_i}(\cdot; p) \right)
\end{equation}

\begin{figure}[ht]
    \centering
    \includegraphics[width=0.9\textwidth]{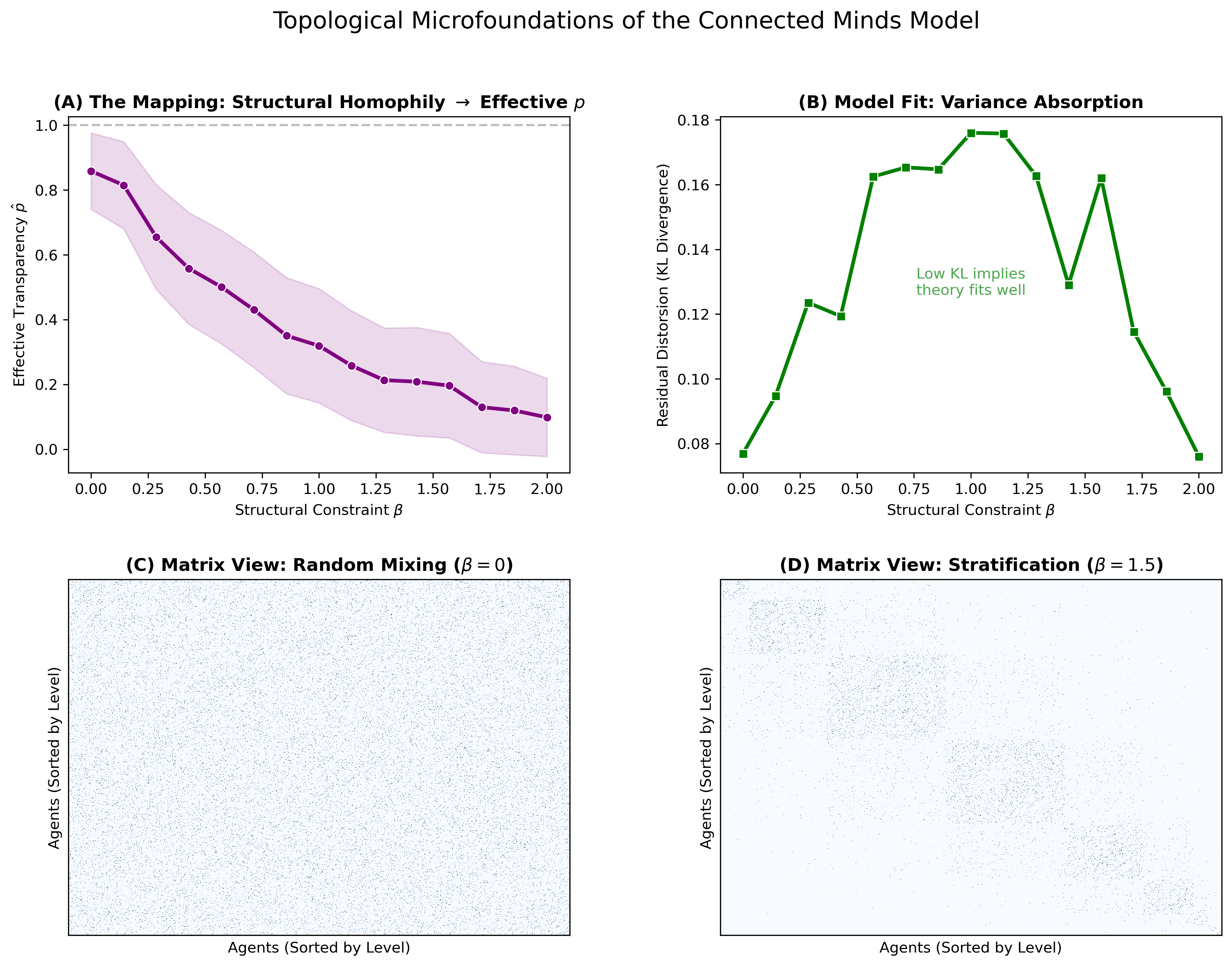}
    \caption{The Topological Microfoundations. (A) The effective transparency $\hat{p}$ decays as structural homophily $\beta$ increases. (B) The Kullback-Leibler divergence between the theoretical belief $g_k(\cdot | \hat{p})$ and the actual empirical neighborhood remains low ($< 0.2$ nats) across the spectrum. This "Variance Absorption" indicates that $p$ is a sufficient statistic for complex topological distortions. (C-D) Adjacency matrices sorted by cognitive level (Spy Plots) reveal the physical transition: from global random mixing (C) to stratified, diagonal clusters (D).}
    \label{fig:topology_mapping}
\end{figure}

\paragraph{Results: The Bridge Between Structure and Theory.}
Figure~\ref{fig:topology_mapping} illustrates the relationship between the structural constraint ($\beta$) and the inferred parameter ($\hat{p}$). We observe three key phenomena:
\begin{itemize}
    \item \textbf{The Random Mixing Limit ($\beta \to 0$):} In Panel A, as $\beta \to 0$, we recover $\hat{p} \approx 0.9$, which is an approximation to the theoretical random‑mixing benchmark ($\hat{p} \approx 1$). Panel C confirms this visually: the adjacency matrix is uniform noise, implying that high-level agents see a representative sample of the population.
    \item \textbf{The Homophily Mapping:} As $\beta$ increases, $\hat{p}$ decays exponentially. Panel D vividly illustrates the mechanism: interactions concentrate along the diagonal. A Level-10 agent essentially sits in a "Level 8-10" cluster, physically cut off from Level-0 agents.
    \item \textbf{Sufficiency of the Parameter (Variance Absorption):} Crucially, Panel B shows that the KL divergence remains uniformly low. Even though the true generative process is a complex exponential decay SBM, the single-parameter belief $g_k(h; p)$ fits the data with high fidelity. This validates our theoretical strategy: we do not need to model the full complexity of $N \times N$ network interactions; the scalar $p$ captures the first-order distortion sufficient for mechanism design.
\end{itemize}
This result justifies the use of $p$ in our mechanism design analysis (Section~\ref{sec:mechanism}): manipulating $p$ is structurally equivalent to altering the homophilic mixing rate of the underlying organization.

\subsection{Connections or Brains?}
When looking for jobs, internships, or even qualifying for a loan, we have all come across an automated screening process that is hard to get past. All of us have heard that someone bypassed the process by ``knowing people'' and went through the process faster and easier. This makes us wonder, can we detect if someone was genuinely qualified to ``win'' in the system or were they simply ``well connected''? For a similar reason, individuals join fraternities and sororities to get more connected, even though it might come at the cost of lower performance in courses or learning \cite{mara2018social}. 

In this experiment, we aim to see if it is possible to distinguish between the connectedness and the sophistication level of a node in Connected Minds model and beauty contest game. The primary objective of this experiment is to demonstrate a ``Ridge of Indeterminacy'' in the likelihood landscape of behavioral data. In traditional Cognitive Hierarchy (CH) models, agents are assumed to have perfect transparency of the population ($p=1$). However, if a population has limited network transparency ($p < 1$), their actions may mimic a standard CH population with a different distribution of cognitive levels.

We test the hypothesis that empirical data generated by a ``smart but isolated'' population (high $\tau$, low $p$) is statistically indistinguishable from a ``less sophisticated but well-connected'' population (low $\tau$, high $p$).

\subsubsection{Experiment Setup}
% \coms{I think we should move it to the main results.}
The simulation environment is built upon two core parameters that govern agent behavior:
\begin{itemize}
    \item $\tau$ (Cognitive Sophistication): The mean of the Poisson distribution $f(k) \sim \text{Poisson}(\tau)$, representing the average number of steps of reasoning agents can perform.
    \item $p$ (Network Transparency): A parameter $p \in (0, 1]$ representing the visibility of the population. A value of $p=1$ recovers the standard CH model, while $p \to 0$ represents an echo chamber where agents' beliefs are concentrated entirely on their immediate cognitive neighbors ($h \approx k-1$), blinding them to the wider distribution of lower-level types.
\end{itemize}
\paragraph{The Action Selection Process}For any level $k$ agent, the theoretical action $s_k$ is calculated recursively based on their beliefs about lower-level players. The belief distribution $g_k(h)$ for an agent of level $k$ regarding the proportion of level $h < k$ players in the network is defined as:
$$g_k(h) = \frac{p^{k-h} \cdot f(h)}{\sum_{j=0}^{k-1} p^{k-j} \cdot f(j)}$$
The agent then selects an action based on the expected action of their opponents, scaled by a multiplier $m$ (typically $2/3$ in beauty contest games):
$$s_k = m \cdot \sum_{h=0}^{k-1} g_k(h) \cdot s_h$$
\paragraph{Data Generation}We initialize a synthetic population of $N = 1000$ agents. We generated synthetic data from three distinct ground-truth populations to test the robustness of our identification strategy: \textbf{Echo-Chambered} $(\tau=1.5, p=0.4)$, \textbf{Moderate} $(\tau=2.5, p=0.7)$, and \textbf{high Transparency} $(\tau=3.0, p=0.8)$. We sample $N$ latent cognitive levels from $\text{Poisson}(\tau)$. We map these levels to a discrete set of actions $\{s_0, s_1, \dots, s_{20}\}$. These actions form our observed dataset $\mathcal{D}$.

\paragraph{Likelihood Estimation and Grid Search}To analyze the identifiability of these parameters, we perform a brute-force grid search across the parameter space $(\tau, p)$. We define a grid where $\tau \in [0.5, 5.0]$ and $p \in [0.1, 1.0]$. For every candidate pair $(\tau', p')$ in the grid, we then calculate the Log-Likelihood ($\mathcal{L}$) of the observed actions $\mathcal{D}$. We assume that observed actions are subject to Gaussian observation noise with a standard deviation $\sigma$. We utilize Log-Likelihood ($\mathcal{L}$) because it numerically stabilizes the product of probabilities for large $N$. Since probabilities are $\le 1$, $\mathcal{L}$ is negative. In terms of $\mathcal{L}$, a \textbf{low absolute value} (closer to 0) indicates a \textbf{good fit} (high probability of observing the data), whereas a \textbf{high absolute value} (highly negative) indicates a \textbf{poor fit}.

%$$P(a_i | \tau, p) = \sum_{k=0}^{K_{max}} P(a_i | s_k(\tau, p)) \cdot P(k | \tau)$$
% Where the likelihood of an individual action $a_i$ given a level $k$ is:
% $$P(a_i | s_k) = \frac{1}{\sqrt{2\pi\sigma^2}} \exp\left( -\frac{(a_i - s_k)^2}{2\sigma^2} \right)$$
% The total log-likelihood for the dataset is:
% $$\mathcal{L}(\mathcal{D} | \tau, p) = \sum_{i=1}^{N} \log \left( \sum_{k=0}^{K_{max}} \frac{f(k|\tau)}{\sqrt{2\pi\sigma^2}} \exp\left( -\frac{(a_i - s_k(\tau, p))^2}{2\sigma^2} \right) \right)$$

\begin{figure}[ht]
    \centering
    \includegraphics[width=0.98\textwidth]{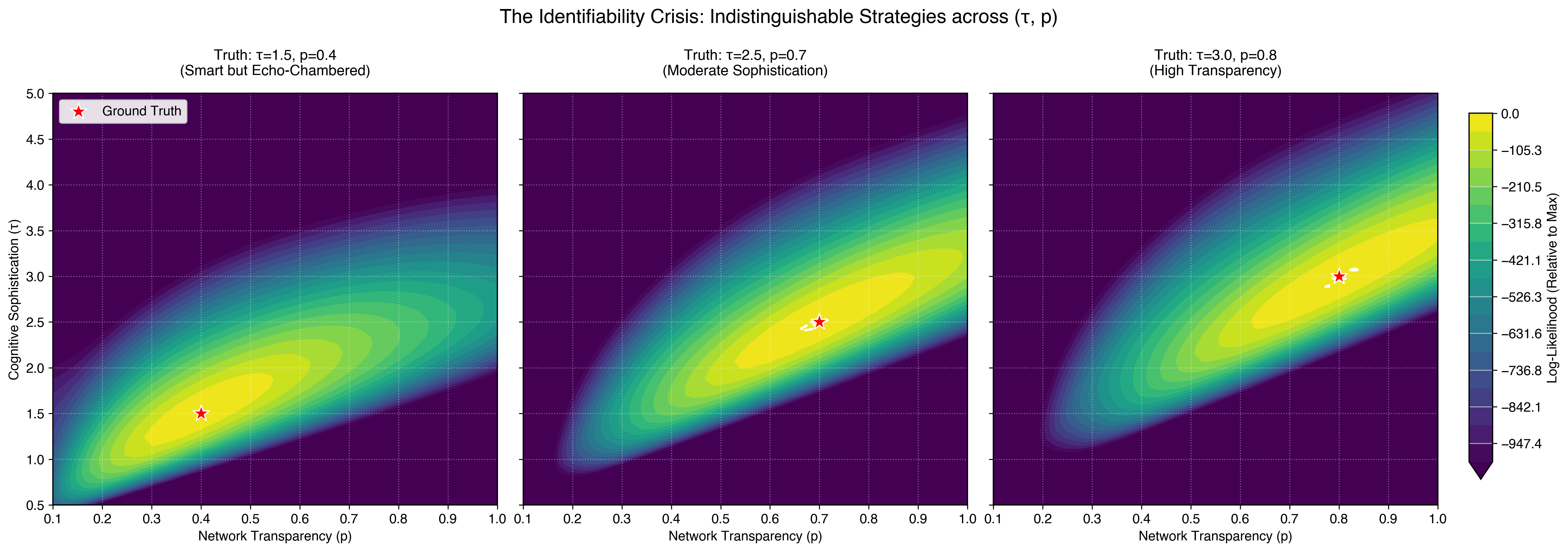}
    \caption{The Identifiability Crisis. We visualize the log-likelihood surface for three different ground-truth scenarios. In each case, a ``Ridge of Indeterminacy'' emerges, prohibiting the separate identification of network bias ($p$) and cognitive level ($\tau$).}
    \label{fig:exp_a}
\end{figure}

\subsubsection{Results and Discussion} 
Figure \ref{fig:exp_a} reveals a ``Ridge of Indeterminacy'' (yellow diagonal and dashed white line). This ridge represents a set of $(\tau, p)$ combinations that yield nearly identical log-likelihood values. The likelihood surface is not peaked at a single point but stretches along a curve where $\tau/p \approx \text{constant}$ and $\tau$ and $p$ trade-off. Because this region is elongated rather than a point, the parameters are jointly unidentifiable

\emph{Implication:} The experiment proves that an empiricist observing only final actions cannot distinguish between a ``smart'' population with perfect network visibility ($p = 1$) and a ``dull'' population hindered by a poor network ($p < 1$).

This has significant implications for social science: if network bias is ignored (assuming $p=1$ by default), researchers will systematically misestimate the cognitive sophistication ($\tau$) of a population. To resolve this ``Identifiability Crisis,'' one would need auxiliary data, such as direct network measurements or longitudinal behavioral shifts.

\subsection{Resolving the Identification Crisis: The ``Info-Shock'' Protocol}
\label{subsec:identification_solution}

The ``Ridge of Indeterminacy'' identified in Figure~\ref{fig:exp_a} poses a significant challenge for econometricians attempting to estimate the model from static field data. The confound $\tau/p \approx C$ implies that without auxiliary information, one cannot distinguish between a sophisticated population in a transparent network and a less sophisticated population in an echo chamber.

However, the Connected Minds model suggests a structural solution available in laboratory settings. Since $\tau$ represents an intrinsic cognitive trait of the population (capacity for reasoning) and $p$ represents an environmental constraint (information availability), we can disentangle the parameters through an \textit{exogenous variation of information topology}.

We propose the \textbf{``Info-Shock'' Identification Strategy}, a within-subject experimental design that breaks the collinearity by enforcing a boundary condition on $p$.

\paragraph{Experimental Design.} The experiment consists of two distinct blocks played by the same cohort of subjects:
\begin{enumerate}
    \item \textbf{Block A (The Transparency Baseline):} Subjects play the Beauty Contest game with \textit{Global Information Feedback}. After each round, the full histogram of all actions from the previous round is displayed.
    \begin{itemize}
        \item \textit{Theoretical Implication:} By providing the global distribution, we exogenously force the visibility parameter to the limit $p \to 1$. As established in Section~\ref{parameter_p}, this effectively removes the network bias kernel.
    \end{itemize}
    \item \textbf{Block B (The Networked Condition):} Subjects play the game with \textit{Local Information Feedback}. Subjects observe only the actions of their immediate neighbors (or a small sample).
    \begin{itemize}
        \item \textit{Theoretical Implication:} The environment reverts to the unknown network parameter $p_{endo} < 1$, reintroducing the bias.
    \end{itemize}
\end{enumerate}

\paragraph{Estimation Strategy.}
Let $\hat{\tau}_{A}$ and $\hat{\tau}_{B}$ be the Maximum Likelihood Estimates (MLE) of the population's sophistication derived from the action data in Block A and Block B, respectively, assuming a standard Poisson CH model.

In Block A, since $p=1$, the estimator recovers the true cognitive parameter:
\begin{equation}
    \hat{\tau}_{A} \approx \tau_{\text{true}}
\end{equation}
In Block B, the network bias inflates the perceived sophistication. Following Theorem~\ref{thm:poissonshiftconvg} (Poisson-Shift Convergence), the observed behavior will converge to a Poisson distribution with the shifted mean:
\begin{equation}
    \hat{\tau}_{B} \approx \frac{\tau_{\text{true}}}{p_{endo}}
\end{equation}
By combining these two regimes, the network transparency parameter $p$ is uniquely identified as the ratio of the estimated sophistications:
\begin{equation}
    p_{endo} = \frac{\hat{\tau}_{A}}{\hat{\tau}_{B}}
\end{equation}

This identification strategy transforms the ``Ridge of Indeterminacy'' from a flaw into a feature: the magnitude of the behavioral shift between the Global and Local conditions serves as a direct measure of the network's opacity. If subjects barely change their behavior ($\hat{\tau}_{A} \approx \hat{\tau}_{B}$), the network is transparent ($p \approx 1$). If behavior shifts drastically towards higher-order strategies in the local condition ($\hat{\tau}_{B} \gg \hat{\tau}_{A}$), it confirms the presence of strong local bias ($p \ll 1$).

\subsection{The Cost of Clarity}\label{sec:num_exp_cost_of_clarity}
With understanding how it is hard to distinguish between connectedness and sophistication, we now aim to understand and quantify how increasing network transparency ($p$) impacts inequality within a population. We investigated the normative implications of transparency. Does providing agents with a clearer view of the population ($p \to 1$) improve social welfare, or does it exacerbate inequality? This section tests the hypothesis that higher transparency enables sophisticated agents to compute significantly more accurate best-responses, while naive agents (Levels 0 and 1) fail to adapt, leading to a widening wealth gap. 

\subsubsection{Experiment Setup}
The simulation utilizes a population of $N = 2,000$ agents with cognitive levels sampled from a Poisson distribution ($\tau = 1.5$). To ensure statistical robustness, the experiment employs Monte Carlo averaging over 1000 independent simulations.
\paragraph{The Payoff Function}To translate strategic success into economic outcomes, we implement an Exponential Payoff Function. In a beauty contest game where the goal is to guess $2/3$ of the average action, the payoff for agent $i$ is defined by its proximity to the target $T$:
$$P_i = \exp(-\lambda |a_i - T|)$$
where $a_i$ is the agent's action and $\lambda$ is the decay rate (set to $0.2$). This function rewards precision exponentially, allowing the strategic advantage of higher-level reasoning to manifest as significant wealth accumulation.
\paragraph{Measuring Inequality} To quantify the resulting wealth distribution, we utilize the Gini Coefficient. In this context, the Gini coefficient measures the statistical dispersion of payoffs across the population. It ranges from $0$ (perfect equality, where everyone receives the same payoff) to $1$ (perfect inequality, where one agent captures the entire reward pool). We calculate the Gini coefficient for the payoff distribution at 50 granular increments of network transparency ($p \in [0.01, 1.0]$) to observe how ``clarity'' alters the economic landscape.

\subsubsection{Results and Discussion}
Contrary to the initial hypothesis that transparency might increase inequality by strictly favoring the super-sophisticated, \emph{the data in the left panel of Figure~\ref{fig:exp_b} reveals a downward trend in the Gini coefficient as $p$ increases}. At low transparency ($p \approx 0$), the Gini coefficient is at its peak ($\sim 0.60$). In this opaque regime, success is concentrated narrowly: only specific intermediate levels (e.g., Level 2) happen to align with the target, while both lower and higher levels miscalculate significantly. As the network becomes more transparent ($p \to 1$), the Gini coefficient drops toward $\sim 0.54$. This reduction in inequality is driven by the \emph{expansion of the high-performing cohort}: rather than success being the monopoly of a single level, transparency allows a broader coalition of agents (Levels 3, 4, and 5+) to converge on the target.

The ``Relative Advantage'' metric (defined as an agent's share of wealth divided by their population percentage) provides a granular view of this transition:
\begin{itemize}
    \item \textbf{The Dilution of Level 2 Dominance:} At low $p$, Level 2 agents hold a massive, effectively monopolistic advantage ($\sim 2.5x$ their fair share) because the opaque network biases beliefs in a way that makes their specific depth of reasoning accidentally optimal. As $p$ increases, Level 2 agents remain successful (staying above parity), but their \emph{relative} dominance erodes as they are forced to share the reward pool with increasingly accurate higher-level agents.
    \item \textbf{The Rise of the Elite (Level 5+):} These agents see their relative advantage climb steadily from near-zero to become the most successful group as $p \to 1$. Crucially, this rise contributes to the lower Gini coefficient because it fills the gap between the ``middle class'' and the ``elite.'' Wealth is no longer spiked at Level 2 but is distributed across a wider gradient of sophisticated players (Levels 3, 4, and 5+).
    \item \textbf{The Stagnant Base (Level 0 \& 1):} As expected, Level 0 and 1 agents remain at the bottom of the hierarchy regardless of transparency, consistently earning near-zero or sub-parity shares.
\end{itemize}

This experiment demonstrates that while Network Transparency shifts the specific locus of maximum advantage toward the ``cognitive elite'' (Level 5+), it fundamentally reduces inequality by breaking the accidental monopoly of moderate thinkers. Clarity effectively raises the cognitive requirements for the top spot, but in doing so, it broadens the base of agents who can successfully compete.

\begin{figure}[ht]
    \centering
    \includegraphics[width=0.98\textwidth]{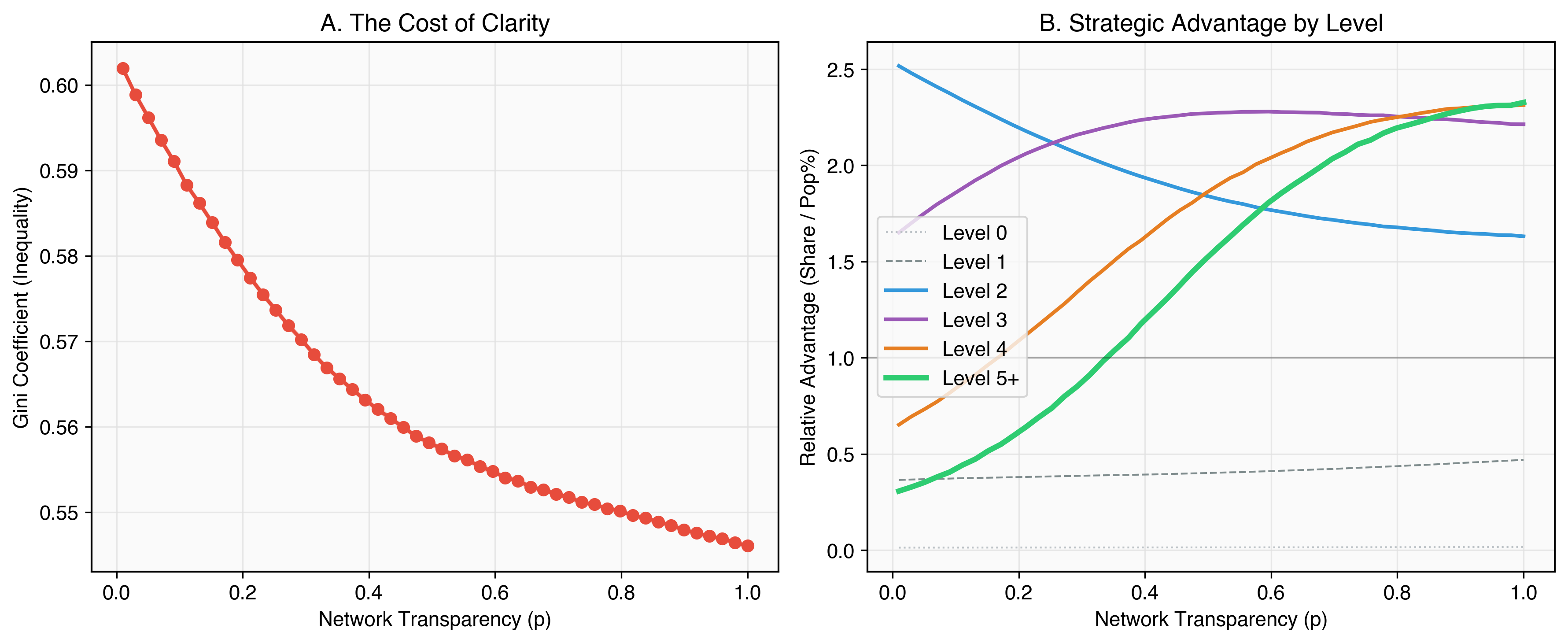}
    \caption{The Cost of Clarity. Panel A (left): Inequality (Gini) lowers with transparency. Panel B (right): The ``Cascade of Dominance'' shows how increasing $p$ successively favors higher cognitive levels (Level 2 $\to$ 3 $\to$ 4 $\to$ 5+), indicating that transparency raises the cognitive requirements for success.}
    \label{fig:exp_b}
\end{figure}

%% file: 92_proofs.tex
\section{Proofs}

\subsection*{Proof of Theorem~\ref{thm:logconcavity}}

A discrete distribution $q(h)$ is log-concave if the sequence satisfies the condition $q(h)^2 \ge q(h-1)q(h+1)$ for all $h$ in the interior of the support. This is equivalent to the concavity of the function $\psi(h) = \ln q(h)$, characterized by the second difference operator $\Delta^2 \psi(h) \le 0$.
Consider the log-belief function for the Connected Minds model:
$$\psi_k(h) = \ln g_k(h) = \ln ( \frac{p^k}{Z_k(p)} p^{-h} f(h) ) = \ln(p^k) - \ln Z_k(p) - h \ln p + \ln f(h)$$

To determine concavity, we compute the second difference with respect to $h$:
$$\Delta^2 \psi_k(h) = \psi_k(h+1) - 2\psi_k(h) + \psi_k(h-1)$$

Substituting the expansion:
$$\Delta^2 \psi_k(h) = [-(h+1)\ln p + \ln f(h+1)] - 2[-h\ln p + \ln f(h)] + [-(h-1)\ln p + \ln f(h-1)]$$
(Note that the constant terms $\ln p^k$ and $\ln Z_k(p)$ vanish upon differencing).
Grouping the linear terms involving $\ln p$:
$$-(h+1)\ln p + 2h\ln p - (h-1)\ln p = \ln p (-h - 1 + 2h - h + 1) = 0$$

The linear tilt $p^{-h}$ contributes a linear term to the log-density, which has a vanishing second derivative (or difference).
Thus, the curvature is determined entirely by the prior: 
$$\Delta^2 \psi_k(h) = \ln f(h+1) - 2\ln f(h) + \ln f(h-1) = \Delta^2 (\ln f(h))$$
By the assumption that $f$ is log-concave, $\Delta^2 (\ln f(h)) \le 0$. Therefore, $\Delta^2 \psi_k(h) \le 0$, proving that $g_k(h)$ is log-concave.

\subsection*{Proof of Proposition~\ref{prop:score_mgf}}

Let $\theta = -\ln p$. The belief $g_k(h)$ can be written in the canonical form of an exponential family: $g_k(h) \propto f(h) e^{\theta h}$. From the properties of the log-partition function $A(\theta) = \ln Z_k(\theta) - k\ln p$, we know that $\frac{\partial \ln Z_k}{\partial \theta} = -\mathbb{E}_{g_k}[k-H]$.

Using the chain rule with $\frac{d\theta}{dp} = -\frac{1}{p}$:
\begin{align}
    \frac{\partial \ln g_k(h)}{\partial p} &= \frac{\partial}{\partial p} ( (k-h)\ln p + \ln f(h) - \ln Z_k(p) ) = \frac{k-h}{p} -( \frac{\partial \ln Z_k}{\partial \theta} \cdot \frac{d\theta}{dp}) \notag\\ 
    &= \frac{k-h}{p} - ( \mathbb{E}_{g_k}[k-H] \cdot \frac{1}{p} )
\end{align}
Factoring out $1/p$ and rearranging terms into the distance from the truncation point $k$ yields the result.

\subsection*{Proof of Proposition~\ref{prop:strategic_sensitivity}}

In an exponential family, the second derivative of the log-partition function with respect to the natural parameter $\theta$ is the variance of the sufficient statistic. Thus, $\frac{\partial^2 \ln Z_k}{\partial \theta^2} = \text{Var}_{g_k}(H)$.

Applying the chain rule to the mean $-\mathbb{E}_{g_k}[k-H] = \frac{\partial \ln Z_k}{\partial \theta}$:
\begin{align}
    - \frac{\partial \mathbb{E}_{g_k}[k-H]}{\partial p} = - \frac{\partial \mathbb{E}_{g_k}[k-H]}{\partial \theta} \cdot \frac{d\theta}{dp} = \text{Var}_{g_k}(H) \cdot ( -\frac{1}{p} )
\end{align}
Since variance is non-negative, this confirms that $\mathbb{E}_{g_k}[H-k]$ is monotonically decreasing in $p$ for all $p > 0$.

\subsection*{Proof of Theorem~\ref{thm:MLRP}}

Let $p_1 > p_2$. Consider the ratio of the probability mass functions for a fixed $k$ and variable $h \in \{0, \dots, k-1\}$:
$$\Lambda(h) = \frac{\frac{1}{Z_k(p_2)} p_2^{k-h} f(h)}{\frac{1}{Z_k(p_1)} p_1^{k-h} f(h)}$$
The objective density $f(h)$ cancels out, isolating the effect of the bias:
$$\Lambda(h) = \frac{Z_k(p_1)}{Z_k(p_2)} \frac{p_2^{k-h}}{p_1^{k-h}} = [ \frac{Z_k(p_1)}{Z_k(p_2)} ( \frac{p_2}{p_1} )^k ] \cdot ( \frac{p_2}{p_1} )^{-h} = C \cdot ( \frac{p_1}{p_2} )^h$$
where $C$ is a positive constant independent of $h$.

Since $p_1 > p_2$, the ratio $\rho = \frac{p_1}{p_2} > 1$.
Thus, $\Lambda(h) = C \rho^h$ is a strictly increasing function of $h$.
An increasing likelihood ratio $\frac{g(h|p_2)}{g(h|p_1)}$ implies that the distribution with the lower parameter ($p_2$) places relatively more weight on higher values of $h$.

By standard stochastic dominance theorems, Likelihood Ratio Dominance implies first-order stochastic dominance (FOSD). Therefore, for any non-decreasing function $u(h)$:
$$\mathbb{E}_{p_2}[u(H)] \ge \mathbb{E}_{p_1}[u(H)]$$

\subsection*{Proof of Theorem~\ref{thm:hierarchyexp}}

We must compare $g_{k+1}$ (supported on $\{0, \dots, k\}$) with $g_k$ (supported on $\{0, \dots, k-1\}$). To establish likelihood ratio dominance ($g_{k+1} \succeq_{LR} g_k$), we must show that for any two types $h_2 > h_1$, the higher type is favored relatively more in $g_{k+1}$ than in $g_k$. That is:$$\frac{g_{k+1}(h_2)}{g_{k+1}(h_1)} \ge \frac{g_k(h_2)}{g_k(h_1)}$$
First, we examine the relative odds for types in the common support $\{0, \dots, k-1\}$.

$$\frac{g_{k+1}(h_2)}{g_{k+1}(h_1)} = \frac{p^{(k+1)-h_2} f(h_2)}{p^{(k+1)-h_1} f(h_1)} = \frac{p \cdot p^{k-h_2} f(h_2)}{p \cdot p^{k-h_1} f(h_1)} = p^{h_1 - h_2} \frac{f(h_2)}{f(h_1)}$$

Comparing this to the odds in $g_k$:
$$\frac{g_k(h_2)}{g_k(h_1)} = \frac{p^{k-h_2} f(h_2)}{p^{k-h_1} f(h_1)} = p^{h_1 - h_2} \frac{f(h_2)}{f(h_1)}$$

Remarkably, the relative odds of any two specific types $h_1, h_2 < k$ are identical for a level-$k$ agent and a level-$(k+1)$ agent. The extra factor of $p$ in the kernel cancels out. Thus, on the common support, the likelihood ratio is constant (equality holds).

Now consider the case where the higher type is the new upper bound, $h_2 = k$, and $h_1$ is any type in the common support ($h_1 < k$).

In the distribution $g_{k+1}$, the probability of $k$ is strictly positive ($g_{k+1}(k) > 0$), so the ratio is positive: $\frac{g_{k+1}(k)}{g_{k+1}(h_1)} > 0$. However, the distribution $g_k$ assigns zero probability to $k$ ($g_k(k) = 0$). Thus: $\frac{g_k(k)}{g_k(h_1)} = 0$.

Since the ratio in $g_{k+1}$ (positive) is strictly greater than the ratio in $g_k$ (zero), the condition holds strictly at the boundary.

Combining these two observations, the relative likelihood of a higher type never decreases as we move from $g_k$ to $g_{k+1}$. It remains constant on the intersection of their supports and strictly increases when moving to the new upper bound. Therefore, the sequence satisfies the definition of being non-decreasing in the likelihood ratio order ($g_{k+1} \succeq_{LR} g_k$), which implies the weaker condition of First-Order Stochastic Dominance.

\subsection*{Proof of Theorem~\ref{thm:poissonshiftconvg}}

Substitute the Poisson PMF into the definition of the belief:
$$g_k(h) = \frac{1}{Z_k} p^{-h} \frac{e^{-\tau} \tau^h}{h!} = \frac{e^{-\tau}}{Z_k} \frac{(\tau/p)^h}{h!}$$
The partition function $Z_k$ sums this quantity over $0 \dots k-1$: $Z_k = e^{-\tau} \sum_{j=0}^{k-1} \frac{(\tau/p)^j}{j!}$. We multiply and divide by $e^{\tau/p}$ to complete the Poisson form inside the sum:
$$Z_k = e^{-\tau} e^{\tau/p} [ e^{-\tau/p} \sum_{j=0}^{k-1} \frac{(\tau/p)^j}{j!} ]$$
The term in the brackets is exactly the Cumulative Distribution Function (CDF) of a Poisson random variable with mean $\lambda = \tau/p$, evaluated at $k-1$. Let us denote this $F_{\tau/p}(k-1)$.
Substituting $Z_k$ back into the expression for $g_k(h)$:
$$g_k(h) = \frac{e^{-\tau} (\tau/p)^h / h!}{e^{-\tau} e^{\tau/p} F_{\tau/p}(k-1)} = \frac{1}{F_{\tau/p}(k-1)} ( \frac{e^{-\tau/p} (\tau/p)^h}{h!} )\equiv \frac{\text{Poisson}(h; \tau/p)}{F_{\tau/p}(k-1)}$$
As $k \to \infty$, the CDF $F_{\tau/p}(k-1) \to 1$.
Therefore, $g_k(h) \to \text{Poisson}(h; \tau/p)$ pointwise and in total variation.

\subsection*{Proof of Theorem~\ref{thm:monotonicityagg}}
Consider two network regimes $p_1 > p_2$ (regime 2 is more opaque). We use strong induction on $k$.
\begin{enumerate}
    \item \textbf{Base Case ($k=0$):} Level-0 agents are non-strategic, so $s_0(p_1) = s_0(p_2) = s_0$. The base case holds with equality.
    \item \textbf{Inductive Step:} Assume $s_h(p_2) \ge s_h(p_1)$ for all $h < k$.
    The best response for agent $k$ is $\beta_k(g_k) = \arg\max \mathbb{E}_{h \sim g_k} [ u(s_k, s_h) ]$.
    The shift from $p_1$ to $p_2$ introduces two effects:
    \begin{itemize}
        \item \emph{Belief Shift (Direct Effect):} By Theorem~\ref{thm:MLRP}, $g_k(\cdot | p_2) \succeq_{FOSD} g_k(\cdot | p_1)$. The agent under $p_2$ places more weight on higher types $h$ (closer to $k$).
        \item \emph{Opponent Escalation (Indirect Effect):} By the inductive hypothesis, for any specific type $h$, the action $s_h(p_2) \ge s_h(p_1)$.
    \end{itemize}
    Combining these, the expected opponent action under $p_2$ is strictly higher:
    $$ \mathbb{E}_{p_2}[s_H(p_2)] = \sum_{h < k} g_k(h|p_2) s_h(p_2) \ge \sum_{h < k} g_k(h|p_2) s_h(p_1) $$
    The inequality follows from the inductive hypothesis. Next, we utilize the Belief Shift. Since $g_k(\cdot|p_2)$ stochastically dominates $g_k(\cdot|p_1)$ and the action sequence $s_h$ is non-decreasing in $h$ (by the cognitive monotonicity hypothesis), it follows that:
    $$ \sum_{h < k} g_k(h|p_2) s_h(p_1) \ge \sum_{h < k} g_k(h|p_1) s_h(p_1) = \mathbb{E}_{p_1}[s_H(p_1)] $$
    Thus, $\mathbb{E}_{p_2}[s_H] \ge \mathbb{E}_{p_1}[s_H]$.
    Under Assumption~\ref{assumption:strat_comp} (strategic complements), the best response function is increasing in the opponent's expected action. Therefore, $s_k(p_2) \ge s_k(p_1)$.
\end{enumerate}
Since $s_k(p_2) \ge s_k(p_1)$ for all $k$, the aggregate action $S^*(p_2) \ge S^*(p_1)$.

\subsection*{Proof of Proposition~\ref{prop:cognitive_FO}}

The objective function is:
$$W(p) = \mathbb{E}_f [ \mathbb{E}_{g_k}[H] ] - c + cp$$
We take the derivative with respect to $p$ to find the critical point. Using the linearity of expectation:
$$\frac{\partial W}{\partial p} = \sum_{k} f(k) \frac{\partial \mathbb{E}_{g_k}[H]}{\partial p} + c$$
Recall Proposition~\ref{prop:strategic_sensitivity}, which established the Variance-Sensitivity Identity: the sensitivity of the mean belief is proportional to the negative variance.
$$\frac{\partial \mathbb{E}_{g_k}[H]}{\partial p} = -\frac{1}{p} \text{Var}_{g_k}(H)$$
Substituting this into the welfare derivative:
$$\frac{\partial W}{\partial p} = \sum_{k} f(k) \left( -\frac{1}{p} \text{Var}_{g_k}(H) \right) + c$$
Setting the first derivative to zero for optimality: $c = \frac{1}{p^*} \sum_{k} f(k) \text{Var}_{g_k}(H)$
and rearranging yields the result:
$p^* = \frac{1}{c} \mathbb{E}_f [ \text{Var}_{g_k}(H) ]$.

\subsection*{Proof of Theorem~\ref{thm:hierarchyofsensitivity}}

From Proposition~\ref{prop:strategic_sensitivity}, the absolute sensitivity is strictly determined by the variance of the belief:
$$|\frac{\partial s_k}{\partial p}| = \frac{\alpha}{p} \text{Var}_{g_k}(H)$$
We analyze the behavior of the normalized metric $\bar{\eta}_k = \frac{\alpha}{pk} \text{Var}_{g_k}(H)$ as $k \to \infty$.

From Theorem \ref{thm:poissonshiftconvg}, we know that as $k \to \infty$, the belief $g_k$ converges to a Poisson distribution with mean $\lambda = \tau/p$. The variance of a Poisson($\lambda$) distribution is equal to its mean $\lambda$. Thus, the absolute variance saturates:
$$\lim_{k \to \infty} \text{Var}_{g_k}(H) = \frac{\tau}{p}$$
However, the normalized elasticity scales by $1/k$. Taking the limit:
$$\lim_{k \to \infty} \bar{\eta}_k(p) = \lim_{k \to \infty} \frac{\alpha}{pk} \left( \frac{\tau}{p} \right) = 0$$
Since $\bar{\eta}_k$ is positive for finite $k$ and converges to $0$, it must eventually be decreasing.
Intuitively, for small $k$, the variance grows as the support expands, increasing sensitivity. However, once $k$ exceeds the effective support of the distribution (the saturation point), the numerator $\text{Var}_{g_k}(H)$ becomes constant while the denominator $k$ continues to grow, driving the normalized sensitivity to zero.